\documentclass[11pt]{article}

\usepackage{amsfonts,amssymb,amsmath,latexsym,enumitem,amsthm}
\usepackage{soul}

\usepackage{graphicx}
\usepackage{xcolor}
\usepackage{float}
\usepackage[ruled,vlined,linesnumbered]{algorithm2e}

\usepackage{authblk}

\newcommand{\FF}{\vspace*{\medskipamount}}
\newcommand{\FFF}{\vspace*{\bigskipamount}}
\newcommand{\B}{\vspace*{-\smallskipamount}}
\newcommand{\BB}{\vspace*{-\medskipamount}}
\newcommand{\BBB}{\vspace*{-\bigskipamount}}

\newcommand{\cA}{\mathcal{A}}

\newcommand{\cE}{\mathcal{E}}

\newcommand{\cG}{\mathcal{G}}
\newcommand{\cH}{\mathcal{H}}

\newcommand{\cM}{\mathcal{M}}

\newcommand{\cO}{O}
\newcommand{\cP}{\mathcal{P}}

\newcommand{\cR}{\mathcal{R}}

\newcommand{\cT}{\mathcal{T}}

\newcommand{\fuzzy}{fuzzy }
\newcommand{\Fuzzy}{Fuzzy }
\newcommand{\blurred}{\fuzzy }

\newcommand{\logO}{\tilde{\cO}}
\newcommand{\Gossip}{\textsc{Gossip}}
\newcommand{\BipartiteGossip}{\textsc{BipartiteGossip}}
\newcommand{\remove}[1]{}

\newcommand{\p}{p}

\newcommand{\ParamConsensus}{\textsc{Parameterized\-Consensus}}

\newcommand{\dk}[1]{{\color{red}{#1}}}

\newcounter{innerlistcounter}

\newcommand{\Paragraph}[1]{\BBB\paragraph{#1}}

\newcommand{\polylog}{\text{polylog }}

\newlength{\pagewidth}
\newlength{\figurewidth}
\makeatletter
\newtheorem*{rep@theorem}{\rep@title}
\newcommand{\newreptheorem}[2]{%
\newenvironment{rep#1}[1]{%
 \def\rep@title{#2 \ref{##1}}%
 \begin{rep@theorem}}%
 {\end{rep@theorem}}}
\makeatother

\newtheorem{theorem}{Theorem}
\newtheorem*{theorem*}{Theorem}
\newreptheorem{theorem}{Theorem}
\newtheorem{lemma}{Lemma}
\newtheorem*{lemma*}{Lemma}

\newtheorem{corollary}{Corollary}

\newtheorem{definition}{Definition}

\newcommand{\ceil}[1]{\left \lceil #1 \right \rceil }
\newcommand{\floor}[1]{\left \lfloor #1 \right \rfloor }

\begin{document}
\title{Improved Communication Complexity of Fault-Tolerant Consensus}

\author{
\vspace*{-.5ex}
MohammadTaghi HajiAghayi  \\
\vspace*{-2ex}
University of Maryland, Maryland, USA. \\
\texttt{hajiagha@cs.umd.edu}
\and
\vspace*{-2ex}
Dariusz R. Kowalski \\
\vspace*{-2ex}
School of Computer and Cyber Sciences,
Augusta University, Georgia, USA. 
\texttt{dkowalski@augusta.edu}	
\and
\vspace*{-2ex}
Jan Olkowski \\
\vspace*{-2ex}
University of Maryland, Maryland, USA. \\
\texttt{olkowski@umd.edu}
}

\date{}

\maketitle


\vspace*{-6.5ex}
\begin{abstract}
Consensus is one of the most thoroughly studied problems in distributed computing, yet there are still complexity gaps that have not been bridged for decades. In particular, in the classical message-passing setting with processes' crashes, since the seminal works of Bar-Joseph and Ben-Or [1998]~\cite{Bar-JosephB98} and Aspnes and Waarts [1996, 1998]~\cite{AspnesW-SICOMP-96,Aspnes-JACM-98} in the previous century, there is still a fundamental unresolved question about communication complexity of fast randomized Consensus against a (strong) adaptive adversary crashing processes arbitrarily online.
The best known upper bound on the number of communication bits is $\Theta(\frac{n^{3/2}}{\sqrt{\log{n}}})$ per process, while the best lower bound is $\Omega(1)$. This is in contrast to randomized Consensus against a (weak) oblivious adversary, for which time-almost-optimal algorithms guarantee amortized $O(1)$ communication bits per process~\cite{GK-SODA-10}.
We design an algorithm against adaptive adversary that reduces the communication gap by nearly linear factor to  $O(\sqrt{n}\cdot\polylog n)$ bits per process,
while keeping almost-optimal (up to factor $O(\log^3 n)$) time complexity $O(\sqrt{n}\cdot\log^{5/2} n)$.

More surprisingly, we show this complexity indeed can be lowered further, but at the expense of increasing time complexity, i.e., there is a {\em trade-off} between communication complexity and time complexity. More specifically, our main Consensus algorithm allows to reduce communication complexity per process to any value from $\polylog n$ to 
$O(\sqrt{n}\cdot\polylog n)$, as long as  Time $\times$ Communication $= O(n\cdot \polylog n)$. Similarly, reducing time complexity requires more random bits per process, i.e., Time $\times$ Randomness $=O(n\cdot \polylog n)$.

Our parameterized consensus solutions are based on a few newly developed paradigms and algorithms for crash-resilient computing, interesting on their own. The first one, called a {\em Fuzzy Counting}, provides for each process a number which is in-between the numbers of alive processes at the end and in the beginning of the counting. 
Our deterministic Fuzzy Counting algorithm works in $O(\log^{3} n)$ rounds and uses 
only $O(\polylog n)$ 
amortized communication bits per process, unlike previous solutions to counting that required $\Omega(n)$ bits.
This improvement is possible due to a new {\em Fault-tolerant Gossip} solution with $O(\log^3 n)$ rounds using only $O(|\cR|\cdot \polylog n)$ communication bits per process, where $|\cR|$ is the length of the rumor binary representation. It exploits distributed fault-tolerant divide-and-conquer idea, in which processes run a {\em Bipartite Gossip} algorithm for a considered partition of processes. To avoid passing many long messages, processes use a family of small-degree compact expanders for {\em local signaling} to their overlay neighbors if they are in a compact (large and well-connected) party, and switch to a denser overlay graph whenever local signalling in the current one is failed. Last but not least, all algorithms in this paper can be implemented in other distributed models such as the congest model in which messages are of length $O(\log n)$.

\vspace*{1ex}

\end{abstract}


\thispagestyle{empty}

\setcounter{page}{0}

\newpage




\section{Introduction}

Fault-tolerant Consensus -- when a number of autonomous processes want to agree on a common value among the initial ones, despite of failures of processes or communication medium -- is 
among
foundation problems in distributed computing.
Since its introduction by Pease, Shostak and Lamport~\cite{PSL}, a large number of algorithms and impossibility results have been developed and analyzed,
applied to solve other problems in distributed computing and systems, and led to a discovery of a number of new
important problems and 
solutions, c.f.,~\cite{AW}.
Despite this persistent effort, we are still far from obtaining even asymptotically optimal solutions in most of the classical distributed~models.

 In particular, in the classical message-passing setting with processes' crashes, since the seminal works of Bar-Joseph and Ben-Or~\cite{Bar-JosephB98} and Aspnes and Waarts~\cite{AspnesW-SICOMP-96,Aspnes-JACM-98} in the previous century, there is still a fundamental unresolved question about communication complexity of randomized Consensus.
More precisely, in this model, $n$ processes communicate and compute in synchronous rounds, by 
sending/receiving messages to/from a subset of processes and
performing local computation.
Each process knows set $\cP$ of IDs of all $n$ processes.
Up to $f<n$
processes may crash accidentally during the computation, which is typically modeled by an abstract adversary
that selects which processes to crash and when,
and additionally -- which messages sent by the crashed processes could reach successfully their destinations. 
An execution of an algorithm against an adversary could be seen as a game, in which the algorithm wants to minimize its complexity measures (such as time and communication bits) while the adversary aims at violating this goal by crashing participating processes. The classical distributed computing focuses on two main types of the adversary: adaptive and oblivious. Both of them know the algorithm in advance, however the former is stronger as it can observe the run of the algorithm and decide on crashes online, while the latter has to fix the schedule of crashes in advance (before the algorithm starts its run).
Thus, these adversaries have different power against randomized algorithms, but same against deterministic ones.

One of the perturbations caused by crashes is that they substantially delay reaching consensus: no deterministic algorithm can reach consensus in all admissible executions within $f$ rounds, as proved by Fisher and Lynch~\cite{FL}, and no randomized solution can do it in $o(\sqrt{n/\log n})$ expected number of rounds against an adaptive adversary, as proved by Bar-Joseph and Ben-Or~\cite{Bar-JosephB98}.
Both these results have been proven (asymptotically) optimal.
The situation gets more complicated if one seeks {\em time-and-communication} optimal solutions.
The only existing lower bound
requires $\Omega(n)$ messages to be sent 
by any algorithm even in some failure-free executions, which gives $\Omega(1)$ bits per process~\cite{AWH}.\footnote{
In this paper we sometimes re-state communication complexity results in terms of the formula {\em amortized per process}, which is the total communication complexity divided by $n$.} 
There exists a {\em deterministic} algorithm with a polylogarithmic amortized number of communication bits~\cite{CK-PODC09}, however deterministic solutions are at least linearly slow, as mentioned above~\cite{FL}.
On the other hand, randomized algorithms running against {\em weak adversaries} are both fast and amortized-communication-efficient, both formulas being $O(\log n)$ or better, c.f., Gilbert and Kowalski~\cite{GK-SODA-10}. 
At the same time, the best randomized solutions against an {\em adaptive adversary} considered in this work requires time $\Theta(\sqrt{n/\log n})$ but large amortized communication  $\Theta(n\cdot \sqrt{n/\log n})$.
In this paper, we show a parameterized algorithm not only improves amortized communication by nearly a linear factor, but also suggests surprisingly that there is no time-and-communication optimal algorithm in this setting.

\B
\Paragraph{Consensus problem.} 
{\em Consensus}
is about making a common decision on some of the processes' input values by every non-crashed process, and is specified by the three requirements:

\BBB
\begin{quote}
	\begin{description}
		\item[\em Validity:]         Only a value among the initial ones may be decided upon.
		\item[\em Agreement:]	No two processes decide on different values.
		\item[\em Termination:] 	Each process eventually decides on some value, unless it crashes.
	\end{description}
\end{quote}

\BBB
\noindent
All the above requirements must hold with probability 1.
We focus on {\em binary consensus}, in which initial values are in $\{0,1\}$.

\section{Our Results and New Tools}

Our main result is
a new consensus algorithm \ParamConsensus$^*$, parameterized by $x$, that achieves {\em any}  asymptotic time complexity between $\logO(\sqrt{n})$\footnote{We use $\logO$ symbol to hide any $\polylog n$ factors.} and $\logO(n)$, while preserving the consensus complexity equation: 
Time 
$\times$
Amortized\_Communication 
$= \cO(n\; \polylog n)$. This is also the first algorithm that makes a smooth transition between a class of algorithms with the optimal running time (c.f., Bar-Joseph's and Ben-Or's \cite{Bar-JosephB98} randomized algorithm that works in $\logO( \sqrt{n} )$ rounds) and the class of algorithms with amortized polylogarithmic communication bit complexity 
(c.f., Chlebus, Kowalski and Strojnowski \cite{CK-PODC09} deterministic algorithm using $\logO(1)$ communication~bits).  

\begin{theorem}[Section~\ref{sec:param-generalized}]
\label{thm:param}
For any $x \in [1,n]$ and the number of crashes $f<n$, \ParamConsensus$^*$ solves Consensus with probability~$1$, in $\logO(\sqrt{nx})$ time and $\logO(\sqrt{\frac{n}{x}})$ amortized bit communication~complexity, whp, using $\logO(\sqrt{\frac{n}{x}})$ random bits per process.
\end{theorem} 

In this section we only give an overview of the most novel and challenging part of \ParamConsensus$^*$, called \ParamConsensus, which solved Consensus if the number of failures $f<\frac{n}{10}$. Its generalization to \ParamConsensus$^*$ is done in Section~\ref{sec:param-generalized}, by exploiting the concept of epochs in a similar way to~\cite{Bar-JosephB98,CK-PODC09}. In short, the first and main epoch (in our case, \ParamConsensus\ followed by {\sc BiasedConsensus} described in Section~\ref{sec:biased-consensus-short}) is repeated $O(\log n)$ times, each time adjusting expansion/density/probability parameters by factor equal to~$\frac{9}{10}$. The complexities of the resulting algorithm are multiplied by logarithmic factor.

\Paragraph{High-level idea of} \hspace*{-0.5em}\ParamConsensus{\bf .}
%
In \ParamConsensus, processes are clustered into $x$ disjoint groups, called super-processes $SP_1,\ldots,SP_x$, of $\frac{n}{x}$ processes each. Each process, in a local computation, initiates its candidate value to the initial value, pre-computes the super-process it belongs to, as well as 
two expander-like overlay graphs which are later use to communicate with other processes.

Degree $\delta$ of both overlay graphs is $O(\log n)$, and correspondingly the edge density, expansion and compactness are selected, 
c.f., Sections~\ref{sub:overlay} and~\ref{sec:param-short}.
One overlay graph, denoted $\cH$, is spanned on the set of $x$ super-processes, 
while copies of the other overlay graph are spanned on the members of each pair of super-processes $SP_i,SP_j$ connected by an edge in $\cH$ (we denote such copy by $SE(SP_i,SP_j)$).

\ParamConsensus{} is split into three phases, c.f., Algorithm~\ref{alg:param:un-conf-gossip} in Section~\ref{sec:param-short}. Each phase~uses~some of the newly developed tools, described later in this section: $\alpha$\textsc{-BiasedConsensus} and \Gossip.
Processes keep modifying their candidate values, starting from the initial values, through different interactions.

\noindent
{\bf\em Using the tools.}
$\alpha$\textsc{-BiasedConsensus} is used for maintaining the same candidate value within each super-process, biasing it towards $0$ if less than a certain fraction $\alpha$ of members prefer $1$; see description in Section~\ref{sec:biased-consensus-short} and \ref{sec:rand-consensus}.
Theorem~\ref{thm:stable-cons} proves that  $\alpha$\textsc{-BiasedConsensus} works correctly in $\logO(\sqrt{n/x})$ time and communication bits per process. 
\Gossip, on the other hand, is used to propagate values between all or a specified group of processes, see description in Section~\ref{sec:gossip-short} and~\ref{sec:gossip-long}.
Theorem~\ref{thm:cheap2-gossip} guarantees that \Gossip\ allows to exchange information between the involved up to $n'$ processes, where $n'\le n$, in time $O(\log^3 n)$ and using $O(\log^6 n)$ communication bits per process (in this application, we are using a constant number of rumors, encoded by constant number of bits).

\noindent
{\bf\em In Phase~1,} super-processes want to flood value $1$ along an overlay graph $\cH$
of super-processes, to make sure that processes in the same connected component of $\cH$ have the same candidate value at the end of Phase 1.
Here by a connected component of graph $\cH$ we understand a maximum connected sub-graph of $\cH$ induced by super-processes of at least $\frac{3}{4}\cdot \frac{n}{x}$ non-faulty processes; we call such super-processes non-faulty.
Recall, that the adversary can disconnect super-processes in $\cH$ by crashing some members of selected super-processes.
To do so, the following is repeated $x+1$ times:
processes in a non-faulty super-process $SP_i$, upon receiving value $1$ from some neighboring non-faulty super-process, make agreement (using {\sc BiasedConsensus}) to set up their candidate value to~$1$ and send it to all their neighboring super-processes $SP_j$ via links in overlay graphs $SE(SP_i,SP_j)$.
One of the challenges that need to be overcome is inconsistency in receiving value $1$ by members of the same super-process, as -- due to crashes -- only some of them may receive the value while others may not. We will show that it is enough to assume threshold $\frac{2}{3}$ in the \textsc{BiasedConsensus}, which together with expansion of overlay graphs $SE(SP_i,SP_j)$ and compactness of $\cH$ (existence of large sub-component with small diameter, c.f., Lemma~\ref{lemma:survival-small-diameter}) guarantee propagation of value $1$ across the whole connected component in $\cH$.
%
It all takes
$(x+1)\cdot (\logO(\sqrt{n/x})+1)=\logO(\sqrt{xn})$ rounds and
$\logO(\sqrt{n/x}+\log n)=\logO(\sqrt{n/x})$ amortized communication per process;
see Section~\ref{sec:phase1} for details.

\noindent
{\bf\em In Phase 2,} 
non-faulty super-processes 
want to estimate the number of 
non-faulty super-processes in the neighborhood of radius $O(\log{x})$ in graph $\cH$. (We know from Phase 1 that whole connected non-faulty component in $\cH$ has the same candidate value.)
In order to do it, they become ``active'' and keep exchanging candidate value $1$ with their neighboring super-processes in overlay graph $\cH$ in stages, until the number of ``active'' neighbors becomes less or equal to a threshold $\delta_x=\Theta(\log x)< \delta$, in which case the super-process becomes inactive, but not more than than $\gamma_x=O(\log x)$ stages.
To assure proper message exchange between neighboring super-processes, \Gossip\ is employed on the union of members of every neighboring pair of super-processes. 
It is followed by \textsc{BiasedConsensus} within each active super-process to let all its members agree if the threshold $\delta_x$ on the number of active neighbors holds.
Members of those super-processes who stayed active by the end of stage $\gamma_x$ (``survived'') conclude that there was at least a certain constant fraction of non-faulty super-processes (each containing at least a fraction of non-faulty members) in such neighborhood in the beginning of Phase 2, and thus they set up variable $\texttt{confirmed}$ to $1$ -- it means they confirmed being in sufficiently large group having the same candidate value and thus they are entitled to decide and make the whole system to decide on their candidate value.
It all takes
$\gamma_x\cdot \logO(\sqrt{n/x}+\log^3 n)\le \logO(\sqrt{xn})$ rounds and
at most $\gamma_x \cdot \delta \cdot \logO(\log^6 n +\sqrt{n/x}) = \logO(\sqrt{n/x})$ amortized communication per process.
See Section~\ref{sec:param-consensus-final} for further details.

\noindent
{\bf\em In Phase 3,} 
we discard the partition into $x$ super-processes. All processes want to learn if there was a sufficiently large group confirming the same candidate value in Phase 2.
To do so, they all execute the \Gossip\ algorithm. Processes that set up variable $\texttt{confirmed}$ to $1$
start the \Gossip\ algorithm with their rumor being 
their candidate value;
other processes start with a null value.
Because super-processes use graph $\cH$ for communication, which in particular satisfies $(\frac{x}{64}, \frac{3}{4}, \delta_{x})$-compactness property (i.e., from any subset of at least $\frac{x}{64}$ super nodes one can choose at least $\frac{3}{4}$ of them such that they induced a subgraph of degree at least $\delta_x$), we will prove that at the end of Phase 2 at least a constant fraction of super-processes must have survived and be non-faulty (i.e., their constant fraction of members is alive). 
Moreover, we show that there could be only one non-faulty connected component of confirmed processes, by expansion of graph $\cH$ that would connect two components of constant fraction of super-processes each (and thus would have propagated value $1$ from one of them to another in Phase 1) -- hence, there could be only one non-null rumor in the \Gossip, originated in a constant fraction of processes.
By property of \Gossip, each non-faulty process gets the rumor and decides on it.
%
It all takes
$\logO(\log^3 n)\le \logO(\sqrt{xn})$ rounds and
at most $\logO(\log^6 n) = \logO(\sqrt{n/x})$ amortized communication per process;
see Section~\ref{sec:phase2} for details.

\noindent
{\bf\em Summarizing,}
each part takes $\logO(\sqrt{xn})$ rounds and
$\logO(\sqrt{n/x})$ amortized communication per process.
Each process uses random bits only in executions of {\sc BiasedConsensus} it is involved to, each requiring $\logO(\sqrt{n/x})$ random bits (at most one random bit per round). The number of such executions is $O(x)$ in Part 1 and $O(\log n)$ in Part 2, which in total gives $\logO(\sqrt{xn})$ random bits~per~process.



\subsection{$\alpha$-Biased Consensus}
\label{sec:biased-consensus-short}

Let us start with the formal definition of $\alpha$-Biased Consensus.
\begin{definition}[$\alpha$-Biased Consensus]
An algorithm solves $\alpha$-Biased Consensus if it solves the Consensus problem and additionally, the consensus value is $0$ if less than $\alpha\cdot n$ initial values of processes are $1$. 
\end{definition}

\noindent In Section~\ref{sec:rand-consensus}, we design an efficient $\alpha$-Biased Consensus algorithm and prove the following:

\begin{theorem}[Section~\ref{sec:rand-consensus}]
\label{thm:stable-cons}
For every constant $\alpha > 0$, there exists an algorithm, called $\alpha$\textsc{-Biased\-Consensus}, that 
solves $\alpha$-Biased Consensus problem with probability $1$, in $\logO(f / \sqrt{n})$ rounds and using $\logO(f / \sqrt{n})$ amortized communication bits whp, for any number of crashes $f < n$.
\end{theorem}

Note that for $f=\Theta(n)$ the algorithm works in $\logO(\sqrt{n})$ rounds and uses $\logO(\sqrt{n})$ communication bits per process.
Observe also that the above result solves classic Consensus as well, and as a such, it is the first algorithm which improves on the amortized communication of 
Bar-Joseph's and Ben-Or's Consensus algorithm \cite{Bar-JosephB98}, which has been known as the best result up for over 20 years. The improvement is by a nearly linear factor $\Theta(n/\log^{13/2} n)$,
while being only $O(\log^3 n)$ away from the absolute lower bound on time complexity (also proved in \cite{Bar-JosephB98}).

\Paragraph{High-level idea of} \hspace*{-0.5em}$\alpha$\textsc{-BiasedConsensus}{\bf .}
%
The improvement 
comes from
replacing a direct communication, in which originally all processes were exchanging their candidate values, by procedure {\sc FuzzyCounting}. This deterministic procedure solves Fuzzy Counting problem, i.e., each process outputs a number between the starting and ending number of active processes, and does it in $O(\log^3 n)$ rounds and with $O(\log^7 n)$ communication bits per process, see Sections~\ref{sec:fuzzy-short},~\ref{sec:goss-fuz-cnt-short} and Theorem~\ref{thm:fuzzy-counting}.

First, processes run 
{\sc FuzzyCounting}
where the set of active processes consists of the processes 
with input value $1$.
Then, each process calculates logical $AND$ of the two values: its initial value and the logical value of formula ``$\texttt{ones} \ge \alpha \cdot n$'', where $\texttt{ones}$ is the number of $1$'s output by the {\sc FuzzyCounting} algorithm. Denote by $x_{p}$ the output of the logical $AND$ calculated by process $p$ -- it becomes $p$'s candidate value.

Next, processes run $O(f / \sqrt{n\log{n}})$ \textit{phases} to update their candidate values 
such that eventually every process keeps the same choice. To do so, in a round $r$ every process $p$ calculates, using the {\sc FuzzyCounting} algorithm, the number of processes with (current) candidate value $1$ and, separately, the number of processes with (current) candidate value $0$, denoted $O_{p}^{r}$ and $Z_{p}^{r}$ respectively. Based on these numbers, process $p$ either sets its candidate value to $1$, if the number $O_{p}^{r}$ is large enough, or it sets it to $0$, if the number $Z_{p}^{r}$ is large, or it replaces it with a random bit, if the number of zeros and ones are close to each other.  

In the Bar-Joseph's and Ben-Or's algorithm the numbers $Z_{p}^{r}$ and $O_{p}^{r}$ were calculate in a single round of all-to-all communication. However, we observe that because processes' crashes may affect this calculation process in 
an arbitrary way (the adversary could decide which messages of the recently crashed processes to deliver and which do not, see Section~\ref{s:model}) and also because messages are simply zeros and ones, this step can be replaced by any solution to Fuzzy Counting. 
In particular, the correctness and time complexity analysis of the original Bar-Joseph's and Ben-Or's algorithm captured the case when an arbitrary subset of 0-1 messages from processes alive in the beginning of this step and a superset of those alive at the end of the step could be received and counted -- and this can be done by our solution to the Fuzzy Counting problem.  


\Paragraph{Monte Carlo version for $f = n - 1$.}
$\alpha$\textsc{-BiasedConsensus} as described above is a Las Vegas algorithm with an expected time complexity $\tau = \logO(\sqrt{n})$, as is the original Bar-Joseph's and Ben-Or's algorithm on which it builds.
However, we can make it Monte Carlo, which is more suitable for application in \ParamConsensus{}, by forcing all processes to stop by time $const \cdot \tau$.
In such case, the worst-case running time will always be
while the correctness (agreement) will hold only~whp.
In order to be applied as a subroutine in the \ParamConsensus{}, we need to add 
one more adjustment, 
so that \ParamConsensus{} could guarantee correctness with probability $1$. 
Mainly, processes which do not decide by time $const \cdot \tau - 2$ initiate a 2-round switch of the {\em whole system of $\cP$ processes} to a deterministic consensus algorithm, that finishes in $O(n)$ rounds and uses $O(\polylog n)$ communication bits per process, e.g., from~\cite{CK-PODC09}.
Such switch between two consensus algorithms has already been designed and analyzed before, c.f.,~\cite{CK-PODC09}, and since this scenario happens only with polynomially small probability, the final time complexity of \ParamConsensus{} will be still $\logO(\sqrt{xn})$ and bit complexity $\logO(\sqrt{n/x})$ per process, both whp and expected.
%

\subsection{Improved Fault-tolerant Gossip Solution}
\label{sec:gossip-short}


The \ParamConsensus{} 
algorithm relies on a new (deterministic) solution to a well-known Fault-Tolerant Gossip problem, in which each non-faulty process has to learn initial rumors of all other non-faulty processes (while it could or could not learn some initial rumors of processes that crash during the execution).
Many solutions to this problem have been proposed (c.f.,~\cite{CK, AlistarhGGZ10}), yet, the best deterministic algorithm given in~\cite{CK} solves Fault-tolerant Gossip in $O(\log^3 n)$ rounds using $O(\log^4 n)$ point-to-point messages amortized per process. However, it requires $\Omega(n)$ amortized communication bits regardless of the size of rumors. We improve this result 
as follows:


\begin{theorem}[Section~\ref{sec:gossip-long}]
\label{thm:cheap2-gossip}
\Gossip{} solves deterministically the Fault-tolerant Gossip problem in $\logO(1)$ rounds using $\logO(|\cR|)$~amortized number of communication bits, where $|\cR|$ is the number of bits needed to encode the~rumors.
\end{theorem}



\Paragraph{High-level idea of} \hspace*{-0.5em}\Gossip{\bf .}
The algorithm implements a distributed divide-and-conquer approach 
that utilizes the \BipartiteGossip{} deterministic algorithm, described in Section~\ref{sec:bipartite-gossip}, in the recursive calls. 
Each process takes the set $\cP$, an initial rumor $r$ and its unique name 
$p \in \cP$
as an input.
The processes split themselves into two groups of size at most $\ceil{n/2}$ each: 
the first $\ceil{n/2}$ processes with the smallest names make the group $\cP_{1}$, while the $n - \ceil{n/2}$ processes with the largest names constitute group $\cP_{2}$. Each of those two groups of processes solves Gossip separately, by evoking the \Gossip{} algorithm inside the group only. The processes from each group know the names of every other process in that group, hence the necessary conditions to execute the \Gossip{} recursively are satisfied. After the recursion finishes, a process in $\cP_{1}$ stores a set of rumors 
$\cR_{1}$ 
of processes from its group, and respectively, a process in $\cP_{2}$ stores a set of rumors 
$\cR_{2}$ 
of processes from its group. Then, the processes solve the Bipartite Gossip problem by executing the \BipartiteGossip{} algorithm on the partition $\cP_{1}$, $\cP_{2}$ and having initial rumors 
$\cR_{1}$ and $\cR_{2}$. 
The output of this algorithm is the final output of the \Gossip.
A standard inductive analysis of recursion and Theorem~\ref{thm:cheap-bi-gossip} stating correctness and $\logO(1)$ time and $\logO(|\cR|)$ amortized communication complexities of \BipartiteGossip{} imply Theorem~\ref{thm:cheap2-gossip}, which proof is deferred to Section~\ref{sec:gossip-long}.

\subsection{Fuzzy Counting}
\label{sec:fuzzy-short}

The abovementioned improvement of algorithm $\alpha$-\textsc{BiasedConsensus} over \cite{Bar-JosephB98} is possible because of designing and employing an efficient solution to a newly introduced Fuzzy Counting problem. 

\begin{definition}[Fuzzy Counting]
An algorithm solves Fuzzy Counting if
each process returns a number between the initial and the final number of active processes. Here, being active depends on the goal of the counting, e.g., all non-faulty processes, processes with initial value $1$, etc.
\end{definition}

\noindent Note that the returned numbers could be different across processes.
In Section~\ref{sec:goss-fuz-cnt-short} we design a deterministic algorithm {\sc FuzzyCounting} and prove the following:

\begin{theorem}[Section~\ref{sec:gossip-long}]
\label{thm:fuzzy-counting}
The \textsc{FuzzyCounting} deterministic algorithm solves the \Fuzzy Counting problem in $\logO(1)$ rounds, using $\logO(1)$ communication bits amortized per process.
\end{theorem}

\Paragraph{High-level idea of} \hspace*{-0.5em}\textsc{FuzzyCounting}{\bf .}
%
%
{\sc FuzzyCounting} uses the \Gossip{} algorithm with the only modification that now we require the algorithm the return the values $\texttt{Z}$ and $\texttt{O}$, instead of the set of learned rumors. We apply the same divide-and-conquer approach. That is, we partition $\cP$ into groups $\cP_{1}$ and $\cP_{2}$ and we solve the problem within processors of this partition.  Let $\texttt{Z}_{1}$, $\texttt{O}_{1}$ and $\texttt{Z}_{2}$, $\texttt{O}_{2}$ be the values returned by recursive calls on set of processes $\cP_{1}$ and $\cP_{2}$, respectively. Then, we use the \BipartiteGossip{} algorithm, described in Section~\ref{sec:bipartite-gossip}, to make each process learn values $\texttt{Z}$ and $\texttt{O}$ of the other group. Eventually, a process returns a pair of values $\texttt{Z}_{1} + \texttt{Z}_{2}$ and $\texttt{O}_{1} + \texttt{O}_{2}$ if it received the values from the other partition during the execution of \BipartiteGossip, or it returns the values corresponding to the recursive call in its own partition otherwise. 
It is easy to observe that during this modified execution processes must carry messages that are able to encode values $\texttt{Z}$ and $\texttt{O}$, thus in this have it holds that $|\cR| = O(\log{n})$.

\remove{
Consider the following modification of Gossip problem. There is given a function $f : 2^{\cP} \rightarrow \{0,1\}^{D}$ to all processes, such that for any two disjoint $A, B \subseteq \cP$ we have that $f(A) + f(B) = f(A \cup B)$ (\textit{additive property}). The processes want to calculate the output of this function on any set $A \subset \cP$ such that every process from $A$ is initially a non-faulty process and every process that returns an output must belong to $A$. For example, setting $f(A) := |A|$, for every $A \subseteq \cP$, in the aforementioned definition corresponds to solving Fuzzy Counting.

To solve this problem, we use the \textsc{CheapGossip} algorithm with the only modification that now we require the algorithm the return the value $f{A}$ for any valid $A$, rather than the set of learned rumors. If the function has the additive property, we can use again \textsc{CheapBipartiteGossip} to make each process learn about the value $f(B)$, which is the output computed by the other recursive call of the \textsc{CheapGossip} algorithm. Using the additive property of the function $f$ it is easy to see the the value $f(A) + f(B) = f(A \cup B)$ is a proper output to the original call. 
Again, because of the additive property, processs can implement the above rules using only $O(d)$ bits per message. The following modification applied to the function $f(A) : = |A|$ gives us a solution to Fuzzy Counting.
}

\subsection{Bipartite Gossip}
\label{sec:bipartite-gossip}

Our \Gossip{} and {\sc FuzzyCounting} algorithms use subroutine {\sc BipartiteGossip} that solves the following (newly introduced) problem.

\begin{definition}
Assume that there are only two different rumors present in the system,
each in at most $\lceil n/2 \rceil$ processes. The partition is known to each process, but the rumor in the other part is not.
We say that an algorithm solves Bipartite Gossip if every non-faulty process learns all rumors of other non-faulty processes in the considered setting.
\end{definition}

Bipartite Gossip is a restricted version of the general {\em Fault-tolerant Gossip} problem, which can be solved in $O(\log^3 n)$ rounds using $O(\log^4 n)$ point-to-point messages amortized per process, but requires $\Omega(n)$ amortized communication bits. In this paper, we give a new efficient deterministic solution to Bipartite Gossip, called {\sc BipartiteGossip}, which, properly utilized, leads to better solutions to Fault-tolerant Gossip and Fuzzy Counting.
More details and the proof of the following result are given in  Section~\ref{sec:bipartite-long}.

\begin{theorem}[Section~\ref{sec:bipartite-long}]
\label{thm:cheap-bi-gossip}
Given a partition of the set of processes $\cP$ 
into two groups $\cP_{1}$ and $\cP_{2}$ of size at most $\ceil{n / 2}$ each, 
deterministic algorithm {\sc BipartiteGossip} solves the Bipartite Gossip problem in $\logO(1)$ rounds and uses $\logO(n \cdot |\cR|)$ bits, where 
$|\cR|$ is the minimal number of bits needed to uniquely encode the two rumors.
\end{theorem}

\Paragraph{{\em High-level idea of}} \hspace*{-0.5em}{\sc BipartiteGossip}{\bf .}
%
If there were no crashes in the system, it would be enough if processes span a bipartite expanding graph with poly-logarithmic degree on the set of vertices $\cP_{1} \cup \cP_{2}$ and 
exchange messages with their initial rumors in $\logO(1)$ rounds. In this ideal scenario the $O(\log{n})$ bound on the expander diameter suffices to allow every two process exchange information, while the sparse nature of the expander graphs contributes to the low communication bit complexity. However, a malicious crash pattern can easily disturb such a naive approach. To overcome this, in our algorithm processes -- rather than communicating exclusively with the other side of the partition -- also estimate the number of crashes in their own group. Based on its result, they are able to adapt the level of expansion of the bipartite graph between the two parts to the actual number of crashes.

The {\em internal} communication within group $P_1$ uses graphs 
from a family of $\Theta(\log{n})$ expanders: $\cG_{\texttt{in}} = \{G_{\texttt{in}}(0), \ldots,  G_{\texttt{in}}(\log{n})\}$, for $t=O(\log{n})$, spanned on the set of processes $\cP_1$ and such that $G_{\texttt{in}}(i) \subseteq  G_{\texttt{in}}(i+1)$, the degree and expansion parameter of the graphs double with the growing index, and the last graph is a clique. They select the next graph in this family every time they observe a significant reduction of non-faulty processes in~their~neighborhood. 
Initially, processes from $\cP_{1}$ span an expander graph $G_{\texttt{in}}(0)$ with $O(\log{n})$ degree on the set~$\cP_{1}$, in the sense that each process in $\cP_1$ identifies its neighbors in the graph spanned on $\cP_1$.
In the course of an execution, each process from $\cP_{1}$ keeps testing the number of non-faulty processes in its $O(\log{n})$ neighborhood in $G_{\texttt{in}}(0)$. If the number 
falls down below some threshold,
the process upgrades the used expanding graph 
by switching to the next graph from the family -- $G_{\texttt{in}}(1)$. 
The process continues testing, and switching graph to the next in the family if necessary, until the end of the algorithm.
The ultimate goal of this 'densification' of the overlay graph is to enable each process' communication with a constant fraction of other alived processes in $\cP_{1}$. Note here that this procedure of adaptive adjustment to failures pattern happens independently at processes in~$\cP_{1}$, therefore it may happen that processes in $\cP_1$ may have neighborhoods taken from different graphs in family $\cG_{\texttt{in}}$.

The {\em external} communication of processes from $\cP_{1}$ with processes from $\cP_{2}$ is 
strictly correlated with their estimation of the number of processes being alive in their $O(\log{n})$ neighborhood in $\cP_{1}$ using expanders in $\cG_{\texttt{in}}$, as described above. Initially, a process from $\cP_{1}$ sends its rumor according to other expander graph 
$G_{\texttt{out}}(0)$
of degree $O(\log{n})$, the first graph in another family of expanders graphs $\cG_{\texttt{out}} = \{G_{\texttt{out}}(0), \ldots,  G_{\texttt{out}}(t)\}$,
for $t=O(\log{n})$, spanned on the whole set of processes $\cP_1 \cup \cP_2$, such that $G_{\texttt{out}}(i) \subseteq  G_{\texttt{out}}(i+1)$, the degree and expansion parameter of the graphs double with the growing index, and the last graph is a clique.
Each time a process chooses a denser graph from family $\cG_{\texttt{in}}$ in the internal group communication, described in the previous two paragraphs, it also switches to a denser graph from family $\cG_{\texttt{out}}$ in the external communication with group $\cP_2$. The intuition is that if a process knows that the number of alive processes in its $O(\log{n})$ neighborhood in $\cP_1$ has been reduced by a constant factor since the last check,
it can afford an increase of its degree in external communication with group $\cP_{2}$ by the same constant factor, as the amortized message complexity should stay the same.

The above dynamic adjustment of internal and external communication degree allows to achieve asymptotically similar result as in the fault-free scenario described in the beginning, up to polylogarithmic factor. More details and the proofs of correctness and performance are in~Section~\ref{sec:bipartite-long}.

\subsection{Local Signalling}

{\sc LocalSignalling} is a specific deterministic algorithm, 
parameterized by a family of overlay graphs provided to the processes.
Processes start at the same time, but may be at different levels -- the level indicates which overlay graph is used for communication.
The name Local Signalling comes from the way it works -- similarly to distributed sparking networks, a process keeps sending messages (i.e., 'signalling')
to its neighbors in its current overlay graph as long as it receives 
enough number of messages from them.
Once a process fails to receive a 
sufficient
number of messages from processes 
that use
the same overlay graph or the previous ones, {\sc LocalSignaling} detects such anomaly and remembers a negative 'not surviving' result (to be returned at the end of the algorithm). Such process does not stop, but rather keeps signaling using less dense overlay graph, in order to help processes at lower level to survive. This non intuitive behavior is crucial in bounding the amortized bit complexity.
The algorithm takes $O(\log n)$ rounds.
Its goal is to leverage the adversary -- if the adversary does not fail many processes starting at a level $\ell$, some fraction of them will survive and exchange messages in $O(\log n)$ time and $O(\polylog n)$ amortized number of communication bits.
More specifically, 
the algorithm run by process $p$ takes as in input: 

\noindent
i) the name of a process $p$ and a set of all processes in the system $\cP$; 

\noindent
ii) an expander-like overlay graph family $\cG = \{G(1), \ldots, G(t)\}$ spanned on $\cP$ such that: $t=O(\log n)$, $G(i) \subseteq  G(i+1)$, the degree and expansion parameters of the graphs double with the growing index, and the last graph is a clique. 
%
Two additional parameters $\delta$ and $\gamma$ describe a diameter and a maximal degree of the base graph $G(1)$, resp. See Section~\ref{sub:overlay}; 

\noindent
iii) the process' (starting) level $\ell_p$, which denotes the index of the graph from family $\cG$ the process currently uses to communicate; and 

\noindent
iv) the message to convey, $r$. 

\noindent
For a given round, let $\cT$ denote a communication graph 
$\cup_{p \in \cP} N_{G(\ell_{p})}(p)$,
that is, a graph with the set of vertices corresponding to $\cP$ and the set of edges determined based on neighbors of each vertex/process $p\in\cP$ from the graph $G(\ell_{p})$ corresponding to the current level $\ell_p$ of process $p$.
Processes exchange messages along this graph, and those who discover that the number of their alived neighbors with the same or higher level $\ell$ is below some threshold, decrease their level by $1$ (i.e., switch their overlay graphs to the previous one in the family). Those who do it at least once during the execution of {\sc LocalSignalling}, which takes $O(\log n)$ rounds, have 'not survived Local Signalling', others 'have survived'.

We will show that if all processes start \textsc{LocalSignaling} at the same time,
those who have survived Local Signalling must have had large-size $O(\log n)$-neighborhoods in graph $\cT$ 
in the beginning of the execution. 
Moreover, they were able to exchange messages with other surviving processes in their $O(\log n)$-neighborhoods, c.f.. Lemma~\ref{lemma:probing}.
We will also prove that the amortized bit complexity 
of the $\textsc{LocalSignaling}$ algorithm is $O(\polylog n)$ per process, c.f., Lemma~\ref{lem:signaling-complexity}. This is the most advanced technical part used in our algorithm -- its full description and detail analysis are given in Section~\ref{sec:local-signaling-short}.

\section{Previous and related work}

\FFF
\Paragraph{Early work on consensus.}
The \emph{Consensus} problem was introduced by Pease, Shostak and Lamport~\cite{PSL}.
Early work focused on {\em deterministic} solutions.
Fisher, Lynch and Paterson~\cite{FLP} showed that the problem is unsolvable in an asynchronous setting, even 
if one process may fail.
Fisher and Lynch~\cite{FL} showed that a synchronous solution requires $f+1$ rounds if up to~$f$~processes~may~crash.


The optimal complexity of consensus with crashes is known with respect to the time and the number of messages (or communication bits) when each among these performance metrics is considered separately. 
Amdur, Weber and Hadzilacos~\cite{AWH} showed that 
the amortized number of messages
per process is at least constant,
even in some failure-free execution.
\remove{%
Dwork, Halpern and Waarts~\cite{DHW} found a solution with $\cO(n\log n)$ messages but requiring an exponential time, and later
Galil, Mayer and Yung~\cite{GMY} developed an algorithm with $\cO(n)$ messages, thus showing that this amount of messages is optimal.
The drawback of the latter solution is that it runs in an over-linear $\cO(n^{1+\varepsilon})$ time, for any $0<\varepsilon<1$.
They also improved the communication to $O(n)$ bits, but the resulting algorithm was exponential in the number of rounds.
Chlebus and Kowalski~\cite{CK} showed that consensus can be
solved in  $\cO(f+1)$ time and with $\cO(n\log^2 f)$ messages
if only the number $n-f$ of non-faulty processors satisfies $n-f=\Omega(n)$.
It was later improved in~\cite{CK-PODC09} to $O(f)$ time and $O(n\; \polylog n)$ number of communication bits.
To summarize, when the number of crashes $f$ could be 
close to
$n$, 
}%
The best deterministic algorithm, given by Chlebus, Kowalski and Strojnowski in \cite{CK-PODC09}, solves consensus in asymptotically optimal time
$\Theta(n)$ and an amortized number of communication bits per process $O(\polylog n)$.

\Paragraph{Efficient randomized solutions against {\em weak adversaries}.}

Randomness proved itself useful to break a linear time barrier for time complexity. 
However, whenever randomness is considered, different types of an adversary generating failures could be considered. 
Chor, Merritt and Shmoys~\cite{CMS} developed constant-time algorithms for consensus against an {\em oblivious adversary} -- i.e., the adversary who knows the algorithm but has to decide which process fails and when before the execution starts.
Gilbert and Kowalski~\cite{GK-SODA-10} presented a randomized consensus algorithm that achieves optimal communication complexity, using $ \mathcal{O}(1) $ amortized communication bits per process 
and terminates in 
$ \mathcal{O}(\log n) $ time with high probability, tolerating up to $ f < n/2 $ crash failures. 

\Paragraph{Randomized solutions against (strong) {\em adaptive adversary}.}
Consensus against an adaptive adversary, considered in this paper, has been already known as more expensive than against
weaker adversaries.
The time-optimal randomized solution to the consensus problem was given by Bar-Joseph and Ben-Or~\cite{Bar-JosephB98}. Their algorithm works in $O(\frac{\sqrt{n}}{\log{n}})$ expected time and uses $O(\frac{n^{3/2}}{\log{n}})$ amortized communications bits per process, 
in expectation. 
They also proved 
optimality of their result with respect to the time complexity, while 
here we substantially improve the communication.

\Paragraph{Beyond synchronous crashes.}
It was shown that more severe failures or asynchrony 
could cause a substantially higher complexity.
Dolev and Reischuk~\cite{DR} and Hadzilacos and Halpern~\cite{HH} proved the $\Omega(f)$ lower bound on the amortized message complexity per process 
of deterministic consensus for {\em (authenticated) Byzantine failures}. 
King and Saia~\cite{KingS11} proved that under some limitation on the adversary and requiring termination only whp, the sublinear expected communication complexity $O(n^{1/2}\polylog{n})$ per process can be achieved even in case of Byzantine failures. 
Abraham et al.~\cite{AbrahamCDNP0S19} showed necessity of such limitations 
to achieve
subquadratic time complexity for Byzantine failures.

If {\em asynchrony} occurs, the recent result of Alistarh
et al.~\cite{AlistarhAKS18} showed how to obtain almost optimal communication complexity $O(n\log{n})$ per process (amortized) if less then $n/2$ processes may fail, which improved upon the previous result $O(n\log^2 {n})$ by Aspnes and Waarts~\cite{AspnesW-SICOMP-96} and is asymptotically almost optimal due to the lower bound $\Omega(n/\log^2 n)$ by Aspnes~\cite{Aspnes-JACM-98}.

\Paragraph{Fault-tolerant Gossip} was introduced by Chlebus and Kowalski~\cite{CK}.
They developed a deterministic algorithm solving 
Gossip
in time $\cO(\log^2 f)$ while using $\cO(\log^2 f)$ amortized messages per process, provided $n-f=\Omega(n)$.
They also showed a lower bound  $\Omega(\frac{\log n}{\log(n\log n)-\log f})$ on the number of rounds
in case $\cO(\polylog n)$ amortized messages are used per process.
In a sequence of papers~\cite{CK,GKS,CK-DISC06},
$O(\polylog n)$ message complexity, amortized per process, was obtained for any $f<n$, while keeping the polylogarithmic time complexity.
Note however that general Gossip requires $\Omega(n)$ communication bits per process for different rumors, as each process needs to deliver/receive at least one bit to all non-faulty processes.
%
Randomized gossip 
against an adaptive adversary is doable
w.h.p. in $O(\log^{2}{n})$ rounds 
using $O(\log^{3}{n})$ communication bits per process, for a constant number of rumors of constant size and for $f < \frac{n}{3}$ processes, c.f.,
Alistarh et al.~\cite{AlistarhGGZ10}.
%

\remove{A graph is said to be \emph{$\ell$-expanding}, or to be an \emph{$\ell$-expander}, if any two subsets of~$\ell$ nodes each are connected by an edge; this notion was first introduced by Pippenger~\cite{Pip}.
	There exist $\ell$-expanders  of the maximum degree $\cO(\frac{n}{\ell}\log n)$, as can be proved by the probabilistic method; such an argument was first used by Pinsker~\cite{Pin}. 
}

\BB
\section{Model and Preliminaries}
\label{s:model}

\B
In this section we discuss the message-passing model in which all our algorithms are developed and analyzed. It is the classic synchronous message-passing model with processes' crashes, c.f.,~\cite{AW,Bar-JosephB98}.

\Paragraph{Processes.}
There are $n$~synchronous processes, with synchronized clocks.
Let $\cP$ denote the set of all processes.
Each process has a unique integer ID in the set $\cP=[n]=\{1,\ldots,n\}$.
The set $\cP$ and its size~$n$ are known to all the processes (in the sense that it may be a part of code of an algorithm); it is also called a {\em KT-1} model in the literature~\cite{Peleg}.

\Paragraph{Communication.}
The processes communicate among themselves by sending messages.
Any pair of processes can directly  exchange messages in a round.
The point-to-point
communication mechanism is assumed to be reliable, in that messages are not lost nor corrupted while in transit.

\Paragraph{Computation in rounds.}
A computation, or an execution of a given algorithm, proceeds in consecutive rounds, synchronized among processes.
By a  \textit{round} we mean such a number of clock cycles that is sufficient to guarantee the completion of the following operations by a process: first, {\em multicasting a message} to an arbitrary set of processes (selected by the process during the preceding local computation in previous round or stored in the starting conditions); second, 
{\em receiving} the sent messages by their (non-faulty) destination processes;
third, performing {\em local computations}.

\Paragraph{Processes' failures and adversaries.}
Processes may fail by crashing. 
A process that has {\em crashed} stops any activity, and in particular does not send nor receive messages.
There is an upper bound $f<n$ on the number of crash failures we want to be able to cope with, which is known to all processes in that it can be a part of code of an algorithm.
We may visualize crashes as incurred by an omniscient {\em adversary} that knows the algorithm and has an unbounded computational power; the adversary decides which processes fail and when. The adversary knows the algorithm and is {\em adaptive} -- if it wants to
make a decision in a round, it knows the history of computation until that point.
However, the adversary does not know the future computation, 
which means that it does not know future random bits drawn by processes.
We do not assume failures to be clean, in the sense that when a process crashes while attempting to multicast a message, then some of the recipients may receive the message and some may not; this aspect is controlled by the adversary.
An {\em adversarial strategy} is a deterministic function, which assigns to each possible history that may occur in any execution some adversarial action in the subsequent round -- i.e., which processes to crash in that round and which of their last messages would reach the recipients.

\Paragraph{Performance measures.}

We consider time and bit communication complexities as performance measures of algorithms.
For an execution of a given algorithm against an adversarial strategy, we define its time and communication as follows.
\emph{Time} is measured by the number of rounds that occur by termination of the last non-faulty process. \emph{Communication} is measured by the total number of bits 
sent in point-to-point messages by termination of the last non-faulty process.
For randomized algorithms, both these complexities are random variables.
Time/Communication complexity of a distributed algorithm is defined as a supremum of time/communication taken over all adversarial strategies, resp.
Finally, time/communication complexity of a distributed problem is an infimum of all algorithms' time/communication complexities, resp.
In this work we present communication complexity in a form of an {\em amortized communication complexity} (per process), which is equal to the communication complexity divided by the number of processes $n$.


\Paragraph{Notation whp.}

We say that a random event occurs \emph{with high probability}, 
or
{\em whp}, if its probability can be lower bounded by $1-\cO(n^{-c})$ for 
a sufficiently large
positive constant $c$.
Observe that when a 
polynomial 
number of events occur whp
each, then 
their union occurs
with high probability~as~well.

\subsection{Overlay Graphs}
\label{sub:overlay}


We review the relevant notation and main theorems assuring existence of specific fault-tolerant compact expanders from~\cite{CK-PODC09}. We will use them as overlay graphs 
in the paper, to specify via which links the processors should send messages
in order to maintain small time and communication~complexities. Some properties of these 
graphs have already been observed in~\cite{CK-PODC09}, however we also prove a new property (Lemma~\ref{lemma:survival-small-diameter}) and use it for analysis of a novel Local Signalling procedure~(Section~\ref{sec:local-signaling-short}).

\Paragraph{Notation.}
Let $G=(V,E)$ denote an undirected graph.
Let $W\subseteq V$ be a set of nodes of~$G$.
We say that an edge $(v,w)$ of $G$ is \emph{internal for~$W$} if $v$ and~$w$ are both in~$W$.
We say that an edge $(v,w)$ of $G$ \emph{connects the sets~$W_1$ and $W_2$}
or \emph{is between $W_1$ and $W_2$}, for any disjoint subsets $W_1$ and~$W_2$ of~$V$, if one of its ends is in~$W_1$ and the other in~$W_2$.
The \emph{subgraph of~$G$ induced by~$W$}, denoted~$G|_W$, is the subgraph of~$G$ containing the nodes in~$W$ and all the edges internal for~$W$.
A node adjacent to a node~$v$ is a \emph{neighbor of~$v$} and the set of all the neighbors of a node~$v$ is the \emph{neighborhood of~$v$}.
$N^i_G(W)$ 
denotes
the set of all the nodes in~$V$ that are of distance at most~$i$ from some node in~$W$ 
in graph~$G$.
In particular, the (direct) neighborhood of~$v$ is denoted~$N_G(v)=N^1_G(v)$.

\Paragraph{Desired properties of overlay graphs.}
Let $\alpha$, $\beta$, $\delta$, $\gamma$ and~$\ell$ be positive integers and $0< \varepsilon <1$ be a real number.
The following definition extends the notion of a lower bound on a node degree:

\BB
\begin{description}
\item[\em Dense neighborhood:]
For a node $v\in V$, a set $S\subseteq N^{\gamma}_G(v)$  is said to be \emph{$(\gamma,\delta)$-dense-neighborhood for~$v$} if each node in~$S\cap N^{\gamma-1}_G(v)$ has at least $\delta$ neighbors in~$S$.
\end{description}
%

\BB
\noindent
We want our overlay graphs to have the following properties, for suitable parameters $\alpha$, $\beta$, $\delta$~and~$\ell$:

\BB
\begin{description}

\B
\item[\em Expansion:]
graph~$G$ is said to be  \emph{$\ell$-expanding}, or to be an \emph{$\ell$-expander}, if any two subsets of~$\ell$ nodes each are connected by an edge.
\B
\item[\em Edge-density:]
graph~$G$ is said to be \emph{$(\ell,\alpha,\beta)$-edge-dense} if, for any set $X\subseteq V$ of \emph{at least} $\ell$ nodes, there are at least $\alpha |X|$ edges internal for~$X$, and for any set $Y\subseteq V$ of \emph{at most} $\ell$ nodes, there are at most $\beta |Y|$ edges internal for~$Y$.
\B			
\item[\em Compactness:]
graph~$G$ is said to be \emph{$(\ell,\varepsilon,\delta)$-compact} if, for any set $B\subseteq V$ of at least $\ell$ nodes, there is a subset $C \subseteq B$ of at least $\varepsilon\ell$ nodes such that 
each node's degree in~$G|_C$ is at least~$\delta$. 
We call any such set~$C$ a \emph{survival set for~$B$}.
\end{description}

\BB
\Paragraph{Existence of overlay graphs.}

Let $\delta, \gamma, k$ be integers such that $\delta = 24\log n$,  $\gamma = 2\log n$ and $25\delta \le k \le \frac{2n}{3}$. 
Let $G(n,p)$ be an Erdős–Rényi random graph
of $n$ nodes, in which each pair of nodes is connected by an edge with probability~$p$, independently over all such pairs.

\begin{theorem}[\cite{CK-PODC09}]
	\label{theorem:overlay-graphs-exist}
	For every $n$ and $k$ such that $25\delta \le k \le \frac{2n}{3}$,
	a random graph 
	$G(n,24\delta/k)$ 
	satisfies all the below properties whp:
	
	\emph{(i)} it is $(k/64)$-expanding, 
\hspace*{4.8em}
	\emph{(iii)} it is $(k,3/4,\delta)$-compact, 

	\emph{(ii)} it is $(k/64,\delta/8,\delta/4)$-edge-dense,
\hspace*{0.5em}
	\emph{(iv)} the degree of each node is between $22\frac{n}{k}\delta$ and $26\frac{n}{k}\delta$.
\end{theorem}

We define an {\em overlay graph} $G(n, k, \delta, \gamma)$ as an arbitrary graph of $n$ nodes fulfilling the conditions of Theorem~\ref{theorem:overlay-graphs-exist}.
Graph $G(n, k, \delta, \gamma)$ can be computed locally (i.e., in a single round) and deterministically by each process.
Specifically, by Theorem~\ref{theorem:overlay-graphs-exist}, 
the class of graphs satisfying the four properties (i) - (iv) is large, therefore any deterministic search in the class of $n$-node graphs, applied locally by each process, returns the same overlay graph $G(n, k, \delta, \gamma)$ in all processes.\footnote{%
Recall that each round contributes~$1$ to the time complexity, no matter of the length of {\em local} computation.}

\begin{lemma}[\cite{CK-PODC09}] \label{lemma:sparse-to-dense}		
If graph~$G=(V,E)$ of $n$ nodes is $(k/64,\delta/8,\delta/4)$-edge-dense then any $(\gamma,\delta)$-dense-neighborhood for a node $v\in V$ has at least $k/64$ nodes, for $\gamma\ge 2\lg n$.
\end{lemma}

\Paragraph{The new property.}

The key new property of overlay graphs with good expansion, edge-\-density and compactness is that survival sets 
in such graphs~have~small~diameters.
\BB
\begin{lemma}
\label{lemma:survival-small-diameter}	
If graph~$G=(V,E)$ of $n$ nodes is
$(\frac{k}{64})$-expanding, 
$(\frac{k}{64},\frac{\delta}{8},\frac{\delta}{4})$-edge-dense
and $(k,\frac{3}{4},\delta)$-compact, 
then for any set $B\subseteq V$ of at least $k$ nodes and for any two nodes $v,w$ from set $C$ being a survival set of $B$, the nodes $v,w$ are of distance at most $2\gamma + 1$ in graph $G_{|C}$, for any $\gamma\ge 2\lg n$. 
\end{lemma}

\section{Parameterized Consensus: Trading Time for Communication}
\label{sec:param-short}

We first specify and analize algorithm $\textsc{ParameterizedConsensus}$, for a given parameter $x \in [1, \ldots, n]$\footnote{
Without loss of generality, 
we may assume that $x$ is a divisor of $n$. If it is not the case, we can always make $\ceil{x}$ groups of size $\ceil{\frac{n}{x}}$, which would not change the asymptotic analysis of the algorithm.
}
and a number of crashes $f<\frac{n}{10}$.
Later, in Section~\ref{sec:param-generalized}, we show how to generalize it to algorithm \ParamConsensus$^*$, which works correctly and efficiently for any number of crashes $f < n$.

\remove{

To the best of our knowledge, this is the first algorithm that makes a smooth transition between a class of algorithms with the optimal running time (c.f., Bar-Joseph's and Ben-Or's \cite{Bar-JosephB98} randomized algorithm that works in $O( \sqrt{\frac{n}{ \log{n} } } )$ rounds) and the class of algorithms with the almost linear bits complexity (c.f., Chlebus's, Kowalski's and Strojnowski's \cite{CK-PODC09} deterministic algorithm that uses $O(n\log^{4}{n})$ communication bits). 
}

\Paragraph{Notation and data structures.}
Let $p \in \cP$ denote the process executing the algorithm, while $b_{p}$ denote $p$'s input bit; $\cP,x,p,b_p$ are the input of the algorithm. 
Let $SP_{1}, \ldots, SP_{x}$ be a partition of the set $\cP$ of processes into $x$ groups of $\frac{n}{x}$ processes each. $SP_i$ is called a \textit{super-process}, and each $p\in SP_i$ is called its {\em member}. We also denote by $SP_{[p]}$ the super-process $SP_i$ to which $p$ belongs. A graph $\cH$ is an overlay graph $G(x,\frac{x}{3}, \delta_{x}, \gamma_{x})$,
which existence and properties are guaranteed in Theorem~\ref{theorem:overlay-graphs-exist} and Lemma~\ref{lemma:survival-small-diameter}, where $\delta_{x} := 24\log{x}, \gamma_{x} := 2\log{x}$. We uniquely identify vertices of $\cH$ with super-processes. We say that two super-processes, $SP_{p}$ and $SP_{q}$, are neighbors if vertices corresponding to them share an edge in $\cH$. For every two such neighbors, we denote by $SE(SP_{p}, SP_{q})$ 
an overlay graph $G(2 \frac{n}{x}, \frac{2n}{3x}, 24\log{\frac{2n}{x}}, 2\log{\frac{2n}{x}})$
which vertices we identify with the set $SP_{p} \cup SP_{q}$. 
($SE(SP_{p}, SP_{q})$ is a short form of \textit{super-edge} between $SP_{p}$ and $SP_{q}$.)
Again, for existence and properties of the above overlay graph we refer to Theorem~\ref{theorem:overlay-graphs-exist} and Lemma~\ref{lemma:survival-small-diameter}. Since the processes operate in $KT$-$1$ model, we 
can assume
that all objects mentioned in this paragraph can be computed locally by any process. 
Below is a pseudo-code of algorithm {\sc ParameterizedConsensus}.

\begin{algorithm}[h!]\label{fig:param-cons}
\SetAlgoLined
\SetKwInput{Input}{input}
\SetKwInput{Output}{output}
\Input{$\cP$, $x$, $p$, $b_p$}
calculate locally $\{SP_{1}, \ldots, SP_{x}\}$, $\cH$ \;
\texttt{candidate\_value} $\leftarrow$ \textsc{ParameterizedConsensus:Phase\_1}($\cP, \{SP_{1}, \ldots, SP_{x}\}, \cH, x, p, b_{p}$)\;
\texttt{confirmed} $\leftarrow$ \textsc{ParameterizedConsensus:Phase\_2}($\cP$, $\{SP_{1}, \ldots, SP_{x}\}$, $\cH$, $x$, $p$)\;
\uIf{$\texttt{confirmed} = 1$}
{
    $\texttt{CandidatesValues} \leftarrow \Gossip(\cP, p, \texttt{candidate\_value})$ \tcc*[r]{\textsc{Phase 3}} \label{alg:param:conf-gossip}
}
\uElse
{
    $\texttt{CandidatesValues} \leftarrow \Gossip(\cP, p, -1)$ \tcc*[r]{\textsc{Phase 3}} \label{alg:param:un-conf-gossip}
}
$\texttt{decision\_value} \leftarrow \text{any value of the set } \texttt{CandidatesValues} \text{ that differs from } -1 $ \;
\Return{\texttt{decision\_value}}
\caption{\textsc{ParameterizedConsensus}}
\end{algorithm}

\Paragraph{\em High-level idea of \ParamConsensus.}
We cluster processes into $x$ disjoint groups (super-processes) of $\frac{n}{x}$ processes each. 
Processes locally compute the super-process they belong to and overlay graphs.
Starting from this point, we 
view the system as a set of $x$  super-processes. 

In the beginning (see line 2 and Section~\ref{sec:phase1} for description of Phase~1), Phase~1 is executed in which super-processes flood value $1$ along an overlay graph $\cH$
of super-processes. The main challenge is to do it in $\logO(\sqrt{xn})$ rounds and $\logO(\sqrt{n/x})$ amortized communication per process whp. 

In Phase 2 (see line 3 and Section~\ref{sec:phase2} for description of Phase~2), super-processes estimate the number of operating super-processes in the neighborhood of radius $O(\log{x})$ in graph $\cH$. 
Members of those super-processes who estimate at least a certain constant fraction (we say that they ``survive''), set up variable $\texttt{confirmed}$ to $1$. The main challenge is to do it in $\logO(\sqrt{n/x})$ rounds and $\logO(\sqrt{nx})$ amortized communication per process whp.

Next, we discard the partition into $x$ super-processes. All processes execute a {\sc Gossip} algorithm. Processes that set up variable $\texttt{confirmed}$ to $1$
start the \Gossip\ algorithm with their initial value being the value of the super-process they belonged to. Other processes start with a null value (-1).
Because super-processes use graph $\cH$ for communication, which in particular satisfies $(x, \frac{3}{4}, \delta_{x})$-compactness property, we will prove that at the end of Phase 2 at least a constant fraction of non-faulty (i.e., their $\frac{3}{4}$ fraction of members are alive) super-processes survive. 
This implies that at least a constant fraction of processes begins the \Gossip\ algorithm with a not-null value. 
Because the not-null value results from a flooding-like procedure of value $1$ (if there is any in the system), we will be able to prove that, eventually, every process gets the same value, since at most a constant number of crashes can occur.

To preserve synchronicity, in the \ParamConsensus\ algorithm we use the Monte Carlo version of \textsc{BiasedConsensus} in both Phase 1 and Phase 2, see discussion in Section~\ref{sec:biased-consensus-short}. However, with a polynomial small probability, in this variant of Consensus some processes may not reach a decision value. To handle this very unlikely scenario, processes who have not decided in a run of \textsc{BiasedConsensus} alarm the whole system by sending a message to every other process. Then, the whole system switches to any  deterministic Consensus algorithm with $O(n)$-time and amortized communication bit complexities (c.f.,~\cite{AW}) and returns its outcome as the final decision. The latter part of alarming and the deterministic Consensus algorithm could use $\Theta(n)$ communication bits, however it happens only with polynomially small probability, see Section~\ref{sec:biased-consensus-short}; therefore, it does not affect the final amortized complexity of the \ParamConsensus\ algorithm whp. For the sake of clarity, we decide to not include this relatively straightforward 'alarm' scheme in the~pseudocodes.

\subsection{Specification and Analysis of Phase 1}
\label{sec:phase1}



\begin{algorithm}[h!] \label{alg:param-cons:phase-1}
\SetAlgoLined
\SetKwInput{Input}{input}
\SetKwInput{Output}{output}
\Input{$\cP$, $\{SP_{1}, \ldots, SP_{x}\}$, $\cH$, $x$, $p$, $b_{p}$}
$\texttt{is\_active} \gets \texttt{true}$ \;
$\texttt{candidate\_value} \gets \frac{2}{3}\textsc{-BiasedConsensus}(p, SP_{[p]}, b_{p})$\; \label{alg:phase-1:zero-stable}
\For{$i \leftarrow 1$ \KwTo $x + 1$}
{
    \uIf{$\texttt{is\_active} = \texttt{true} \And \texttt{candidate\_value} = 1$}
    {
        $\texttt{candidate\_value} \leftarrow \frac{2}{3}\textsc{-BiasedConsensus}(p, SP_{[p]}, \texttt{candidate\_value})$ \;\label{alg:phase-1:first-stable}
    }
    \uElse
    {
        stay silent for $y=\logO(\sqrt{\frac{n}{x}})$ rounds \;
    }
    \uIf{$\texttt{is\_active} = \texttt{true} \And \texttt{candidate\_value} = 1$}
    {
        \ForEach{super-process $SP_{j}$ being a neighbor of $SP_{[p]}$ in $\cH$ \label{alg:phase-1:begin-comm}}
        {
            send $1$ to every member of $SP_{j}$ which is a neighbor of $p$ in $SE(SP_{[p]}, SP_{j})$ \; \label{alg:phase-1:direct-comm}
        }\label{alg:phase-1:end-comm}
        $\texttt{is\_active} \leftarrow \texttt{false}$ \;
    }
    \uIf{$p$ received a message containing $1$ in the previous round}
    {
        $\texttt{candidate\_value} \leftarrow 1$ \;
    }
}
\Return{
$\frac{1}{3}\textsc{-BiasedConsensus}(p, SP_{[p]}, \texttt{candidate\_value})$}
\caption{\textsc{ParameterizedConsensus:Phase\_1}}
\end{algorithm}

\paragraph{\em High-level idea.}
In the beginning, the members of every super-process agree on a single value among their input values. Once this is done, super-processes start a flooding 
procedure
navigated by an overlay graph $\cH$. 
$\cH$ should be an expander-like, regular graph with good connectivity properties, but a small degree of at most $O(\log{x})$. Intuitively, this can guarantee that regardless of the crash pattern there will exist a connected component, of a size being a constant fraction of all vertices, in $\cH$ consisting of super-processes that are still operating. The flooding processes is a sequential process of $O(x)$ \textit{phases}. A single super-process \textit{communicates}, that means it sends value $1$ to all its neighbors in $\cH$, in at most one phase only; either in the first phase, if the value its members agreed on in the beginning is $1$; or in the very first phase after the super-process received value $1$ from any of its neighbors in $\cH$. 
End of the flooding process encloses the Phase~1~of~the~algorithm.

Once members of a super-process get value $1$ for the first time, their use {\sc BiasedConsensus} to agree if value $1$ has been received or not. It is necessary due to crashes during the flooding process, yet it is not easy to implement with low amortized bit complexity. A pattern of crashes can result in some members of a super-process receiving value $1$ and some other not. One can require all members to execute {\sc BiasedConsensus} in each phase, but this will blow up the amortized bit complexity to $\logO(x\sqrt{\frac{n}{x}})$ whp. We in turn, propose to execute {\sc BiasedConsensus} only among this members who received value $1$ in the previous communication round (i.e. line~\ref{alg:phase-1:direct-comm}) and use the stronger properties of Biased Consensus problem to argue that the number of calls to the {\sc BiasedConsensus} algorithm will not be to large.

\Paragraph{Analysis.} 
Recall, that we say that a super-process \textit{communicates} with another super-process if \textit{any of its members} executes lines~\ref{alg:phase-1:begin-comm}-\ref{alg:phase-1:end-comm} of the Algorithm~\ref{alg:param-cons:phase-1}.
Trivially, from the Algorithm~\ref{alg:param-cons:phase-1} we get that each member of a super-process executes line~\ref{alg:phase-1:direct-comm} at most once, since if the line is executed then variable $\texttt{is\_active}$ will be changed to $\texttt{false}$, but the next lemma shows that we can expect more: members of a super-process preserve synchronicity in communicating with other members.

\begin{lemma} \label{lem:phase-1-sync-comm}
For every $i \in [x]$, there is at most one iteration of the main loop in which $SP_{i}$ communicates with any other super-process.
\end{lemma}
\begin{proof}
Let us fix any super-process $SP_{i}$ and consider the \textit{first} round $r$ in which that super-processes communicates with another one. If such round does not exist the lemma holds. Now, if a member $p$ of the super-process $SP_{i}$ executes line~\ref{alg:phase-1:direct-comm}, it must have its variables $\texttt{is\_active}$ and $\texttt{candidate\_value}$ set to $\texttt{true}$ and $1$, respectively. In particular this means, that $p$ had to execute line~\ref{alg:phase-1:first-stable} before it reached line~\ref{alg:phase-1:direct-comm} in this iteration. Otherwise, its value $\texttt{candidate\_value}$ would be $0$. This, let us conclude that the $\frac{2}{3}$-\textsc{BiasedConsensus} algorithm executed in line~\ref{alg:phase-1:first-stable} returned value $1$. Since, we used a Biased Consensus algorithm, see Theorem~\ref{thm:stable-cons}, we have that at least $\frac{2}{3}|SP_{i}|$ members started the synchronous execution of line~\ref{alg:phase-1:first-stable}. It easily follows, that each of this members either became faulty or executed line~\ref{alg:phase-1:direct-comm} later in the same iteration of the main loop. If a member executed line~\ref{alg:phase-1:direct-comm}, it sets its variable \texttt{is\_active} to \texttt{false} and stays idle for the rest of the algorithm run, in particular it does not participate in any future run of the $\frac{2}{3}$-\textsc{BiasedConsensus} algorithm. The same holds if a member became a faulty process. It gives us that at most $\frac{1}{3}|SP_{i}|$ members of $SP_{i}$ can participate with value other than $0$. Since they always execute the $\frac{2}{3}$-\textsc{BiasedConsensus} algorithm, thus according to Theorem~\ref{thm:stable-cons}, the result will always be $0$. This proves that $SP_{i}$ will not communicate in any other round than $r$. 
\end{proof}

\begin{lemma} \label{lem:phase-1-same-return}
For every $i \in [x]$, members of a non-faulty super-process $SP_{i}$ return the same value in $\textsc{Phase\_1}$.
\end{lemma}
\begin{proof}
Each member returns its decision based on the result of the Biased variant of a Consensus algorithm executed on the set of all members of its super-process, thus according to Theorem~\ref{thm:stable-cons}, each member must return the same value.
\end{proof}

\noindent Recall, that we defined a super-process \textit{non-faulty} if in the end of the  $\texttt{ParameterizedConsensus}$ algorithm at least $\frac{3}{4}$ of its members have not been crashed. In particular, the number of operating members is at least $\frac{3n}{4x}$ in every phase of the algorithm.

\begin{lemma}
There are no two non-faulty super-processes that are connected by an edge in $\cH$ but they members return different $\texttt{decision\_values}$ in the end of \textsc{Phase\_1}.
\end{lemma}
\begin{proof}

Assume contrary, that there exist two super-processes $SP_{1}$, $SP_{2}$ that are connected by an edge in $\cH$, such that members $SP_{1}$ return value $0$, but members of $SP_{2}$ return $1$. For the super-process $SP_{2}$, the returned value is calculated based on values of variables \texttt{candidate\_value} of its members, which are in turn calculated based on output of the $\frac{2}{3}$-\textsc{BiasedConsenus} algorithm from line~\ref{alg:phase-1:first-stable} (optionally, it can be a result of the run of this algorithm in line~\ref{alg:phase-1:zero-stable}, but than line~\ref{alg:phase-1:first-stable} is executed as well). Therefore, there must be at least iteration of the main loop in which members of $SP_{2}$ agreed on $1$ in the line~\ref{alg:phase-1:first-stable}. Observe, that in the same iteration these members must communicate, i.e. execute line~\ref{alg:phase-1:direct-comm}. According to Lemma~\ref{lem:phase-1-sync-comm} there is at most one such iteration. Let $k$ be the number of this iteration.

Suppose that $k = x + 1$. This means that members of $SP_{2}$ have set variable \texttt{candidate\_value} to $1$ based on received messages from members of a neighboring super-process $SP_{a}$. We observe, this messages must be received in the preceding iteration, that is in the iteration $k - 1$, super-process $SP_{3}$ communicated with $SP_{2}$. By Lemma~\ref{lem:phase-1-sync-comm} there is at most on such iteration for $SP_{3}$, which in turn gives us that some other super-process must communicate with $SP_{3}$ in the iteration $k - 2$. By backwards induction and the fact the each super-process communicates at most once, we get a chain of \textit{distinct} super-processes $P_{2}, SP_{3}, \ldots, SP_{x + 2}$, such that for $j \in [3, x + 2]$ super-process $SP_{j}$ communicated with super-process $SP_{j - 1}$ in the iteration $r - j + 2$ of the main loop. The chain consists of $x$ distinct super-processes, thus $SP_{1}$ must belong to it. This in turn means that there exists an iteration of the main loop, in which members of $SP_{1}$ communicated by sending message $1$ to other super-processes, i.e. they executed lines~\ref{alg:phase-1:begin-comm}-\ref{alg:phase-1:end-comm}. 
If the communication happened in this iteration, by a retrospective reasoning, we can conclude the $\frac{2}{3}$-\textsc{BiasedConsensus} in line~\ref{alg:phase-1:first-stable}, that proceeds the communication must result in value $1$. From the property of the $\alpha$-Biased Consensus, we have that at least $\frac{2}{3}$ fraction of members of $SP_{1}$ started the line~\ref{alg:phase-1:first-stable} having value \texttt{candidate\_value} set to $1$. Because $SP_{1}$ is non-faulty, thus at least $\frac{2}{3} - \frac{1}{4} \ge \frac{1}{3}$ fraction of members from $SP_{1}$ remain non-faulty to the end of $\textsc{Phase\_1}$. Therefore, the outcome of the $\frac{1}{3}-$\textsc{BiasedConsensus} algorithm must be $1$ which is a contradiction with the assumption that members of $SP_{1}$ have set the variable \texttt{decision\_value} to $0$.

Consider now the case where $k < x + 1$. We shall show that members of $SP_{2}$ send message $1$ to sufficiently many members of $SP_{1}$ in iteration $k + 1 \le x + 1$ to influence the value their return. 
Provided that $SP_{2}$ communicates in round $k$, we observe that members must also execute line~\ref{alg:phase-1:first-stable}, and to make the communication possible, the $\frac{2}{3}-$\textsc{BiasedConsensus} must return value~$1$. From the properties of $\alpha-$Biased Consensus, we get that at least $\frac{2}{3}$ fraction of members of $SP_{2}$ started the Consensus algorithm with $\texttt{candidate\_value}$ set to $1$. Since $SP_{2}$ is non-faulty, thus at least $\frac{2}{3} - \frac{1}{4} = \frac{5}{12}$ fraction of members of $SP_{2}$ took part in sending messages to members of $SP_{1}$ in lines~\ref{alg:phase-1:begin-comm}-\ref{alg:phase-1:end-comm}.
The graph $SE(SP_{2}, SP{1})$ satisfies properties of Theorem~\ref{theorem:overlay-graphs-exist}; in particular, by Lemma~\ref{lemma:sparse-to-dense}, we get that at least $\frac{11}{12}$ members of $SP_{1}$ received message $1$ in the iteration $k$. 
Using the fact that $SP_{1}$ is non-faulty, we argue that at least $\frac{11}{12} - \frac{1}{4} = \frac{2}{3}$ members of $SP_{1}$ participated in the run of the $\frac{2}{3}$-\textsc{BiasedConsensus} algorithm in the next $k + 1$ iteration. These processes preserved the variable $\texttt{candidate\_value}$ set to $1$ since sufficiently many members started and finished the run of the $\frac{2}{3}$-\textsc{BiasedConsensus} algorithm. These members eventually take part in the execution of the $\frac{1}{3}$-\textsc{BiasedConsensus} algorithm at the end of \textsc{Phase\_1}. Since their start the execution having $\texttt{candidate\_value}$ set to $1$ and they do not crash, thus the result of the $\frac{1}{3}$-\textsc{BiasedConsensus} must be $1$. This proves the lemma.
\end{proof}

\noindent 
From the previous lemma we can immediately conclude.

\begin{lemma} \label{lem:stable_groups}
Members of each connected component of $\cH$ formed by a non-faulty super-processes return the same $\texttt{decision\_values}$ in the end of \textsc{Phase\_1}.
\end{lemma}

\begin{lemma} \label{lem:phase-1}
The \textsc{Phase\_1} part of the $\textsc{ParameterizedConsensus}$ algorithm takes $\logO(x\sqrt{n / x})$ rounds and uses $\logO(n \sqrt{n / x} \log{n})$ bits whp.
\end{lemma}
\begin{proof}
The upper bound on the number of rounds follows from the observation that in each iteration of the main loop of the $\textsc{Phase\_1}$ algorithm, the execution of $\frac{2}{3}$-\textsc{BiasedConsensus} in line~\ref{alg:phase-1:first-stable} takes $\logO(\sqrt{n / x})$ rounds whp, by Theorem~\ref{thm:stable-cons}, and every other instruction is just a single round communication. By applying union bound over all iterations of the main loop, we get the the total running time of the $\textsc{Phase\_1}$ algorithm is $\logO(x\sqrt{n / x})$ whp.

To get the correct upper bound on the bit complexity of the algorithm we will first bound the number iterations in which members of a super-process $SP_{[p]}$ call the $\frac{2}{3}$-\textsc{BiasedConsensus} algorithm in line~\ref{alg:phase-1:first-stable}. 
Observe, that this line is executed only if both variables $\texttt{is\_active}$ and $\texttt{candidate\_value}$ are set to $\texttt{true}$ and $1$ respectively. Also, from the pseudocode of $\textsc{Phase\_1}$ it follows that before the current iteration ends either $\texttt{is\_active}$ will be set to $\texttt{false}$, or $\texttt{candidate\_value}$ will be changed to $0$. Thus, to execute line~\ref{alg:phase-1:first-stable} in any future iteration, members of $SP_{[p]}$ must receive a message $1$ from a member of a neighbor of $SP_{[p]}$ in $\cH$, since this is the only way to re-set $\texttt{candidate\_value}$ to $1$. But the crux is, that this cannot happen more than $\delta_{x}$ number of times. Indeed, in Lemma~\ref{lem:phase-1-sync-comm} we proved that members of each super-processes send messages to other members at most once. Moreover, they do this in the very same round. Since $SP_{p}$ has no more than $\delta_{x}$ neighbors in $\cH$ we get that $\frac{2}{3}$-\textsc{BiasedConsensus} algorithm in line~\ref{alg:phase-1:first-stable} can be executed at most this number of times among members of $SP_{[p]}$. Combining this fact with the complexity bounds given in Theorem~\ref{thm:stable-cons}, we get the the total number of bits used for all runs of the $\frac{2}{3}$-\textsc{BiasedConsensus} algorithm among members of $SP_{p}$ is $O(\delta_{x} \cdot (n / x)\sqrt{n / x} \log^{4}{n / x})$ whp.

Members may communicate in only one other way, by sending single bits in line~\ref{alg:phase-1:direct-comm}. However, by entering the if clause containing this line, a member must change its variable $\texttt{is\_active}$ from $\texttt{true}$ to $\texttt{false}$. Since this operation is irrevocable, the line~\ref{alg:phase-1:direct-comm} may be executed at most one by each process. For communication between super-processes, $SP_{[p]}$ and $SP_{[q]}$, we used sparse graphs $SE(SP_{[p]}, SP_{[q]})$ of degree $\delta_{n / x} = O(\log(n / x)$. Also the degree of $SP_{[q]}$ in $\cH$ is $\delta_{x}$ which gives us the in total members of $SP_{[p]}$ use $O((n / x) \cdot \log(n / x) \cdot \delta_{x})$ bits for the second type of communication.

Summing the two above estimation over all super-processes gives us the claimed upper bound on the number of bits used by $\texttt{Phase\_1}$.
\end{proof}

\subsection{Specification and Analysis of Phase 2}
\label{sec:phase2}


\remove{We say that a super-process \textit{survives} \textsc{Phase\_2} if it is non-faulty and all its members return value $1$ in $\textsc{Phase\_2}$. Observe, that this definition is consistent because members of a super-processes change their variable $\texttt{is\_active}$ based on outputs of the $\alpha\text{-stable-}\textsc{BiasedConsensus}$ algorithm. Thus the members always return the same bit in the end of $\textsc{Phase\_2}$.
}

\begin{algorithm}[h!]
\SetAlgoLined
\SetKwInput{Input}{input}
\SetKwInput{Output}{output}
\Input{$\cP$, $\{SP_{1}, \ldots, SP_{x}\}$, $\cH$, $x$, $p$}
\uIf{$\frac{3}{4}\textsc{-BiasedConsensus}(p, SP_{[p]}, 1)$ = 1}
{
    $\texttt{is\_active} \leftarrow \texttt{true}$ \label{alg:phase-2:is-active-ini}
}
\uElse
{
    $\texttt{is\_active} \leftarrow \texttt{false}$ \tcc*[r]{stage $i$}
}
\For{$i \leftarrow 1$ \KwTo $\gamma_{x}$}
{
    \uIf{$\texttt{is\_active} = \texttt{true}$}
    {
        $SN \gets \emptyset$ \;
        \ForEach{super-process $SP_{j}$ being a neighbor of $SP_{[p]}$ in $\cH$}
        {
            $N_{j}\gets$ \Gossip($SP_{[p]} \cup SP_{j}, p, p$) \; \label{alg:phase-2:cheap-gossip}
            $SN \leftarrow SN \cup N_{j}$ \;
        }
        \lIf{$|SN| > \delta_{x}$}
        {
            $\texttt{many\_superprocesses} \leftarrow 1$
        }
        \lElse
        {
            $\texttt{many\_superprocesses} \leftarrow 0$
        }
        $\texttt{survived} \leftarrow \frac{2}{3}\textsc{-BiasedConsensus}(p, SP_{[p]}, \texttt{many\_superprocesses})$ \;
        \If{$\texttt{survived} = 0$}
        {
            $\texttt{is\_active} \leftarrow \texttt{false}$
        }
    }
}
\Return{\texttt{is\_active}\tcc*[r]{a bit indicating whether $p$'s super-process survived}}
\caption{\textsc{ParameterizedConsensus}:Phase\_2}
\end{algorithm}

\paragraph{\em High-level idea.}

In Phase 2,
non-faulty super-processes estimate the number of operating super-processes in the neighborhood of radius $O(\log{x})$ in graph $\cH$. 
Those who estimate at least a certain constant fraction, set up variable $\texttt{confirmed}$ to $1$.
In order to achieve that, each super-process keeps signaling all its neighbors in $\cH$ in $\gamma_x=O(\log{x})$ stages until at least a constant fraction of them signaled its activity in preceding stage. A super-process that has been signaling during all stages is said \textit{to survive}. We will prove 
that, thanks to suitably chosen connectivity properties of $\cH$, at least a constant fraction of super-processes survives. Members of these super-processes will influence the final decision of the whole system in the following Phase 3.

\Paragraph{Analysis.}

\begin{lemma}\label{lem:fraction-survives}
At least $\frac{1}{2}$ super-processes are non-faulty and survive $\textsc{Phase\_2}$ of the \textsc{Parameterized\-Consensus} algorithm.
\end{lemma}
\begin{proof}
The lemma follows from the connectivity properties of the graph $\cH$. We say that a super-process becomes \textit{inactive} whenever its members set the variable $\texttt{is\_active}$ to \texttt{false}. Observe that this definition is consistent since the variable is always an output of the $\alpha$-\textsc{BiasedConsensus} algorithm.

Let $S$ be the set of non-faulty super-processes. Because adversary can crash at most $\frac{1}{10}$ processes, thus $|S| > \frac{6}{10}x$. First, we see that every super-process belonging $S$ starts \textsc{Phase\_2} with $\texttt{is\_active}$ being set to $1$, since during the entire execution it has at least $\frac{3}{4}$ fraction of non-faulty processes. The compactness property of $\cH$ ensures that there exists a survival set $C \subset S$, $|C| > \frac{5}{6}|S| = \frac{1}{2}x$. Because each super-process of $C$ has at least $\delta_{x}$ neighbors in $C$ (i.e. other super-processes connected with it by an edge in $C$), thus members of super-processes from $C$ receive at least $\delta_{x}$ different rumors when they execute Fault-tolerant Gossip algorithm in line~\ref{alg:phase-2:cheap-gossip}, in each iteration of the main loop. Therefore, every super-process from $C$ survives \textsc{Phase\_2}. 
\end{proof}

\begin{lemma} \label{lem:phase-2}
The \textsc{Phase\_2} part of the $\textsc{ParameterizedConsensus}$ algorithm takes $\logO(\sqrt{n / x})$ rounds and uses $\logO(n\sqrt{nx})$ bits whp.
\end{lemma}
\begin{proof}
We separately calculate running time of each sub-algorithm used in $\textsc{Phase\_2}$. According to Theorem~\ref{thm:stable-cons} each run of the $\alpha$-\textsc{BiasedConsensus} algorithm on a group consisting of members of a single super-process lasts $\logO(\sqrt{n/x})$ whp. The \Gossip\ algorithm is executed on a group of processes that has size $2(n / x)$ and from Theorem~\ref{thm:cheap2-gossip} we conclude that this single execution has running time $\logO(1)$. Since in \textsc{Phase\_2} we repeat the aforementioned subroutines $\delta_{x}$ times we get that total running time is $\logO(\sqrt{n / x})$ whp.s

Each execution of the $\alpha$-\textsc{BiasedConsensus} algorithm costs $\logO(\frac{n}{x}\sqrt{\frac{n}{x}})$ bits whp. Members of a single super-process execute $\gamma_{x} + 1 = 2\log(x) + 1$ instances of the $\alpha$-$\textsc{BiasedConsensus}$ algorithm. Since there is $x$ super-processes in total, thus the $\textsc{Phase\_2}$ algorithm uses $\logO(n\sqrt{n / x})$ bits for evokes of the $\alpha-$Biased Consensus algorithm. The other communication bits processes generate only by participating in the \Gossip\ algorithm with members of neighbors of its super-process, c.f executing line~\ref{alg:phase-2:cheap-gossip}. Members of a single super-process participate in $\delta_{x}$ parallel executions of Gossip in a single iteration, since this is the degree of every vertex in $\cH$. Each execution of the \Gossip\ algorithm uses $\logO(\frac{n}{x}|\cR|)$ bits according to Theorem~\ref{thm:cheap2-gossip}, where $|\cR|$ denotes the number of bits needed to encoded all initial rumors. However, in our case there are two initial rumors of size $O(\log{n})$ - the two identifiers of the super-processes sharing an edge in $\cH$. 
Since we have only $\gamma_{x}$ iterations and $x$ super-processes, we get the total number of bits used for the evokes of \Gossip\ algorithm is $\logO(n)$. Therefore, the number of bits used by processes in \textsc{Phase\_2} is $\logO(n\sqrt{n / x} + n) = \logO(n\sqrt{nx})$ whp, as claimed.
\end{proof}

\subsection{Analysis of algorithm~\textsc{ParameterizedConsensus}}
\label{sec:param-consensus-final}

\begin{lemma} \label{lemma:super-are-same}
The value \texttt{candidate\_value} is the same among all members of super-processes that survived \textsc{Phase\_2}.
\end{lemma}
\begin{proof}
If a super-process $SP_{i}$ has survived \textsc{Phase\_2}., than it had been continuously communicating with at least $\delta_{x}$ other super-processes for at least $\gamma_{x}$ stages. From the similar reasoning to this in proof of Lemma~\ref{lemma:probing} and the choice of $\cH$ to be $(x, \frac{3}{4}, \gamma_{x})$-compact, we conclude that $SP_{i}$ must belong to a connected component in $\cH$ consisting of survived super-processes of size $\frac{1}{2}x$ at least in the end of \textsc{Phase\_2}. Observe, that every super-processes that survived must be a non-faulty processes in the end of \textsc{Phase\_1}, c.f. line~\ref{alg:phase-2:is-active-ini} of \textsc{Phase\_2}.
According to Lemma~\ref{lem:stable_groups}, members of all super-processes belonging to the same connected component of non-faulty super-processes in $\cH$ share the same $\texttt{candidate\_value}$. Since there can be at most one connected component of non-faulty super-processes of size $ > \frac{1}{2}x$, we see that all members of survived super-processes have the same value of the variable $\texttt{candidate\_value}$.

\end{proof}

\begin{lemma} \label{lem:param-3-cond}
The algorithm~\textsc{ParameterizedConsensus} satisfies validity, agreement and termination conditions.
\end{lemma}
\begin{proof}
The validity conditions follows from the fact, that processes always manipulate only the values that were given to them as in input.

For the agreement condition, we first observe that by Lemma~\ref{lemma:super-are-same}, all members of super-processes that survived share the same value of the \texttt{candidate\_value} variable. Also, we observe that only members of survived super-processes feed the execution of the \Gossip\ algorithm with an initial rumor different than $-1$. In particular, we have that the \Gossip\ in lines~\ref{alg:param:conf-gossip} and~\ref{alg:param:un-conf-gossip} (Phase 3) is executed with at most two different initial rumors, $-1$ and $\texttt{candidate\_value}$ of members of survived super-processes. By Lemma~\ref{lem:fraction-survives}, a fraction of at least $\frac{1}{2}$ super-processes survived. The 
total number of members belonging to this set is $\frac{1}{2}\cdot\frac{n}{x}\cdot x = \frac{1}{2}n$. Since, no more than $\frac{1}{10}n$ processes crash in the course of \textit{the whole} execution, at 
least $\frac{1}{2}n - \frac{1}{10}n = \frac{2}{5}n > 0$ will be non-faulty in the end of \Gossip\ algorithm execute in line~\ref{alg:param:un-conf-gossip} of the main algorithm. By Theorem~\ref{thm:cheap2-gossip}, we conclude that every non-faulty process learns the value of $\texttt{candidate\_value}$ that was the input to this \Gossip\ algorithm. This gives the agreement condition.

The termination follows immediately, given the fact that \textsc{Phase\_1} and \textsc{Phase\_2} terminate with probability $1$,\footnote{Recall here, that we obtained the probability $1$ of termination by applying the Monte Carlo version of \textsc{BiasedConsensus} algorithm.} by Lemma~\ref{lem:phase-1} and Lemma~\ref{lem:phase-2}. The \Gossip\ algorithm is deterministic and by Theorem~\ref{thm:cheap2-gossip} it terminates in $\logO(1)$ rounds.%
\end{proof}

\begin{theorem}
\label{thm:param-restricted}
For any $x \in [1,n]$ and any number of crashes $f<\frac{n}{10}$, \ParamConsensus\ solves Consensus with probability~$1$,
in $O(\sqrt{nx}\; \polylog{n})$ time
and $O(\sqrt{\frac{n}{x}}\; \polylog{n})$ amortized bit communication~complexity, whp, using $O(\sqrt{\frac{n}{x}}\; \polylog{n})$ random bits per process.
\end{theorem}
\begin{proof}
By Lemma~\ref{lem:param-3-cond} we already know that the \textsc{ParameterizedConsensus} algorithm is a solution to the Consensus problem.

By Lemma~\ref{lem:phase-1} and Lemma~\ref{lem:phase-2} we get the time and bit complexity of $\textsc{Phase\_1}$ and $\textsc{Phase\_1}$. By Theorem~\ref{thm:cheap2-gossip}, we have that a single execution of a \Gossip\ algorithm takes $\logO(1)$ rounds and $\logO(1)$ communication bits amortized per process, given that there can be only two different rumors of size $\logO(1)$ each, as we argued in Lemma~\ref{lem:param-3-cond}. These bounds together give us the desired complexity of the \textsc{ParameterizedConsensus} algorithm.

A single run of the $\alpha$-\textsc{BiasedConsensus} algorithm on members of a super-processes generates $\logO(\frac{n}{x}\sqrt{\frac{n}{x}})$ random bits, since each member generates at most one random bit per every round of the algorithm, see Section~\ref{sec:biased-consensus-short}. Since, the processes execute at most $\logO{x}$ runs of the $\alpha$-\textsc{BiasedConsensus} algorithm, thus the total number of random bits used is $\logO(n\sqrt{\frac{n}{x}})$ which implies $\logO(\sqrt{\frac{n}{x}})$ amortized random bit complexity.
\end{proof}

\subsection{Generalization to any number of failures.}
\label{sec:param-generalized}

In this subsection we highlight main ideas that generalize the \textsc{ParameterizedConsensus} algorithm to work in the presence of {\em any} number of crashes $f < n$. We call the resulting algorithm \ParamConsensus$^*$.
We exploit the concept of epochs in a similar way to~\cite{Bar-JosephB98,CK-PODC09}. In short, the first and main epoch (in our case, \ParamConsensus\ followed by {\sc BiasedConsensus} described in Section~\ref{sec:biased-consensus-short}) is repeated $O(\log n)$ times, each time adjusting expansion/density/probability parameters by factor equal to~$\frac{9}{10}$. The complexities of the resulting algorithm are multiplied by logarithmic factor. More details are given below.

Consider a run of the \textsc{ParameterizedConsensus} algorithm, as described and analyzed in previous sub-sections. Let us analyze the state of the system at the end of \textsc{ParameterizedConsensus} algorithm if more than $\frac{n}{10}$ crashes have occurred. In the end, there exist two group of processes, those that have $\texttt{decision\_value}$ set to $-1$ (i.e., the last \Gossip\ has not been successful in their case), and those who have $\texttt{decision\_value}$ set to a value from $\{0, 1\}$. Observe, that if at most $\frac{n}{10}$ processes were faulty, then we already proved in Theorem~\ref{thm:param-restricted} that the first of these sets would be empty and there could be only one value in $\{0,1\}$ taken by alive processes. Thus, we can extend the run of the \textsc{ParameterizedConsensus} by an execution of $\frac{1}{2}$\textsc{-BiasedConsensus} among members of each super-processes, separately for different super-processes, to make them agree if there exists a member of the super-process who had received a null value in the last \Gossip\ execution. 
A single run of \textsc{ParameterizedConsensus} followed by the run of $\frac{1}{2}$-\textsc{BiasedConsensus} is called an \textit{epoch}.
Based on the 
output of the $\frac{1}{2}$\textsc{-BiasedConsensus}, the members of each super-process decide whether they keep the agreed candidate value as decision final value and stay idle in the next epoch, or they continue to the next epoch.
There are three key properties here. First, because the decision of entering next epoch is made based on an output to Biased Consensus, it is consistent among members of a single super-process. Second, in the good scenario, i.e., when only less than $\frac{n}{10}$ processes crashed, every process will start the run of the $\frac{1}{2}$\textsc{-BiasedConsensus} with the same value, yet different than a null-value. From validity condition, all processes stay idle.
Third, a non-faulty super-process at the end of Phase~2 actually implies that there was a majority of non-faulty other super-processes in its $O(\log n)$ neighborhood, regardless of the number of failures (c.f., Lemma~\ref{lemma:probing} -- thus, only one value in $\{0,1\}$ can be confirmed in the whole system as long as at least one process remains alive, whp.

In the next epoch, super-processes that are not idle, repeat the 
\textsc{ParameterizedConsensus} algorithm, but tune its parameters to adjust to the larger number of crashes (i.e., smaller fraction of alive processes). They use: 
\begin{itemize}
\item
a graph $\cH_{1}$, instead of $\cH$, which is roughly $\frac{10}{9}$ denser (i.e a graph $G()$) compared to graph $\cH$ used in the previous epoch, 
\item
new threshold $\alpha_{1} := \frac{2}{3} \cdot \frac{9}{10}$ for evoking  \textsc{BiasedConsensus} algorithm, 
\item
they loose the parameter in the definition of a non-faulty super-process by a factor of $9/10$. 
\end{itemize}
In general, processes repeats this process of 'densification' in subsequent $\Theta(\log{n})$ epochs. Eventually, one of this epochs must be successful, otherwise the number of crashed process would exceed $n / (1 / 10)^{\Theta(n)} > n$. On the other hand, each time we 'densify' graph $\cH$, i.e., we take an overlay graph $\cH_i$ from the family of overlay graphs as defined in Section~\ref{sub:overlay} but with expansion and density parameters adjusted by factor $\left( \frac{9}{10} \right)^i$,
we are guaranteed that only a fraction of previously alive processes execute the next epoch. As density and expansion parameters in the family of overlay graphs are inversely proportional, 
we conclude
that in each epoch the amortized bit complexity stays at the same level of $O(\sqrt{\frac{n}{x}})$. Therefore, in cost of multiplying both, the time complexity and the amortized bit complexity by a factor of $\Theta(\log{n})$, we are able to claim Theorem~\ref{thm:param}.

\begin{reptheorem}{thm:param}[Strengthened Theorem~\ref{thm:param-restricted}]
For any $x \in [1,n]$ and the number of crashes $f<n$, \ParamConsensus$^*$ 
solves Consensus 
with probability~$1$,
in $O(\sqrt{nx}\; \polylog{n})$ 
time
and $O(\sqrt{\frac{n}{x}}\; \polylog{n})$ 
amortized bit communication~complexity, whp, using $O(\sqrt{\frac{n}{x}}\; \polylog{n})$ random bits per process.
\end{reptheorem}


\section{Randomized $\alpha$-Biased Consensus} 
\label{sec:rand-consensus}
The $\alpha$-$\textsc{BiasedConsensus}$ algorithm generalizes and improves the $\textsc{SynRan}$ algorithm of Bar-Joseph and Ben-Or~\cite{Bar-JosephB98}.  For this part, we purposely use the same notation as in~\cite{Bar-JosephB98} for the ease of comparison.

First, processes run Fuzzy Counting (i.e. use the \textsc{FuzzyCounting} algorithm from Section~\ref{sec:goss-fuz-cnt-short}) where the set of active processes consists of this processes which the input value to the $\alpha$-Biased Consensus is $1$. Then, each process calculates logical $AND$ of the two values: its initial value and $\texttt{ones} \ge \alpha \cdot n$, where $\texttt{ones}$ is the number of $1$'s output by the Fuzzy Counting algorithm. Denote $x_{p}$ the output of the logical $AND$ calculated by process $p$.

In the following processes solves an $\alpha$-Biased Consensus on $x_p$. Each process $p$ starts by setting its current choice $b_{p}$ to $x_{p}$. The value $b_{p}$ in the end of the algorithm indicates $p$'s decision. Now, processes use $O(f / \sqrt{n\log{n}})$ \textit{phases} to update their values $b_{p}$ such that eventually every process keeps the same choice. To do so, in a round $r$ every process $p$ calculates the number of processes that current choice is $1$ and the number of processes that current choice is $0$, denoted $O_{p}^{r}$ and $Z_{p}^{r}$ respectively. Based on these numbers, process $p$ either sets $b_{p}$ to $1$, if the number $O_{p}^{r}$ is large enough; or it sets $b_{p}$ to $0$, if the number $Z_{p}^{r}$ is large; or it replaces $b_{p}$ with a random bit, if the number of zeros and ones are close to each other.  
In Bar-Joseph's and Ben-Or's the numbers $Z_{p}^{r}$ and $O_{p}^{r}$ were calculate in a single round all-to-all of communication. However, we observed that because processes' crashes may affect this calculation process in almost arbitrary way, this step can be replaced by any solution to Fuzzy Counting. That holds, because Fuzzy Counting exactly captures the necessary conditions that processes must fulfill to simulate the all-to-all communication, that is it guarantees that candidate values of non-faulty processes are included in the numbers $O_{p}^{r}$ and $Z_{p}^{r}$ calculated by any processor $p$. Thus, rather than using all-to-all communication, our algorithms utilizes the effective \textsc{FuzzyCounting} algorithm where active processes are those who have their current choice equal $1$. The output of this algorithm serves as the number $O_{p}^{r}$, while the number $Z_{p}^{r}$ is just $n - O_{p}^{r}$. For the sake of completeness, we also provide the pseudocode of the algorithm. We conclude the above algorithm in the Theorem~\ref{thm:stable-cons}.

\begin{algorithm}[h]\label{fig:group-cons}
\SetAlgoLined
\SetKwInput{Input}{input}
\SetKwInput{Output}{output}
\Input{$\cP$, $p$, $b_{p}$, $\alpha$}
\Output{a consensus value}
\lIf{\ul{$\textsc{FuzzyCounting}(\cP, p, b_{p}) > \alpha \cdot |\cP|$}}{\ul{$x_{p} \leftarrow b_{p} \ \& \ 1$}}
\lElse{\ul{$x_{p} \leftarrow 0$}}
$r:=1$; $N_{-1}^r = N_0^r = n$; $\texttt{decided} = FALSE$ \;
\While{$TRUE$}
{
	\ul{participate in \textsc{CheapCounting} execution with input bit being set to $b_p$; let $O_p^r$, $Z_p^r$ be the numbers of ones and zeros (resp.) returned by \textsc{CheapCounting}}\;
	$N_p^r = Z_p^r + O_p^r$\;
	\If{$(N_p^r < \sqrt{n/ \log n })$}
	{
	    send $b_p$ to all processes, receive all messages sent to $p$ in round $r+1$\;
	    implement a deterministic protocol for $\sqrt{n/ \log n}$ rounds\;
	}
	\If{$\texttt{decided} = TRUE$}
	{
	    \texttt{diff} $=$ $N_p^{r-3}$ $N_i^r$\;
        \lIf{(\texttt{diff} $\leq N_p^{r-2}/10$)}{STOP}
        \lElse{\texttt{decided} $= FALSE$}
	}
	\lIf{$O_p^r > (7N_p^r-1)/10$}{$b_p = 1$, decided $= TRUE$}
	\lElseIf{$O_p^r > (6N_p^r-1)/10$}{$b_p = 1$}
	\lElseIf{$Z_p^r = 0$}{$b_p = 1$}
	\lElseIf{$O_p^r < (4N_p^r-1)/10$}{$b_p = 0$, decided $= TRUE$}
	\lElseIf{$O_p^r < (5N_p^r-1)/10$}{$b_p = 0$}
	\lElse{set $b_p$ to $0$ or $1$ with equal probability}
	
    $r := r + 1$\;
}

\Return{$b_{p}$}
\caption{$\alpha$-\textsc{BiasedConsensus}. The part in which our algorithm differs from the \textsc{SynRyn} algorithm from~\cite{Bar-JosephB98} algorithm is underlined.}
\end{algorithm}

\begin{reptheorem}{thm:stable-cons}
The $\alpha$\textsc{-BiasedConsensus} algorithm solves $\alpha$-Biased Consensus with probability $1$. The algorithm has expected running time $O(f / \sqrt{n} \cdot \log^{5/2}n)$ and the expected amortized bit complexity $O(f / \sqrt{n} \cdot \log^{13/2}n)$, for any number of crashes $f < n$.
\end{reptheorem}

\noindent Setting $\alpha := \frac{1}{2}$ we get a better randomized solution to classic Consensus problem.
\begin{corollary}
The $\frac{1}{2}$-\textsc{BiasedConsensus} algorithm is a solution to Consensus. The algorithm satisfies agreement and validity with probability 1, has expected running time $O(f / \sqrt{n} \cdot \log^{5/2}n)$, and the expected amortized bit complexity $O(f / \sqrt{n} \cdot \log^{13/2}n)$, for any number of crashes $f < n$.
\end{corollary}

\Paragraph{Monte Carlo version.}
The original algorithm $\alpha$-\textsc{BiasedConsensus} has the expected running time $O(\sqrt{n}\log^{13/2}n)$. However, we can force all processes to stop by that time multiplied by a constant.
In such case, the worst-case running time will be always $\logO(\sqrt{n})$ while the correctness (agreement) will hold only~whp.

\section{Gossip and Fuzzy Counting} 
\label{sec:goss-fuz-cnt-short}

In this section we design and analyze an algorithm, called \Gossip\, which, given a set of processes $\cP$, solves the Gossip problem in $\logO(1)$ rounds and uses $\logO(|\cR|)$ communication bits amortized per process, where $|\cR|$ is the number of bits needed to encode initial rumors of all processes. A small modification of this algorithm will result in a solution to the Fuzzy Counting problem with the same time and only logarithmically larger bit complexity. 

\subsection{Bipartite Gossip}
\label{sec:bipartite-long}

We start by giving a solution to Gossip problem in a special case, called \textit{Bipartite Gossip}, in which processes are partitioned into two groups $\cP_{1}$ and $\cP_{2}$ each of size $\ceil{n / 2}$ at most. Processes starts with at most two different initial rumors $r_{1}$ and $r_{2}$ such that processes of each group share the same initial rumor. The partition and the initial rumor is assumed to be an input to the algorithm. The goal of the system is still to achieve Gossip.

\Paragraph{High level idea of algorithm}\hspace*{-1em} {\sc BipartiteGossip}{\bf .} 
If there were no crashes in the system, it would be enough if processes span a bipartite expanding graph with poly-logarithmic degree on the set of vertices $\cP_{1} \cup \cP_{2}$ and for $\logO(1)$ rounds exchange messages with their initial rumors. In this ideal scenario the $O(\log{n})$ bound on the expander diameter suffices to allow every two process exchange information, while the sparse nature of the expander graphs contributes to the small bit complexity. However, a malicious crash pattern can easily disturb such naive approach. To overcome this, in our algorithm processes will adapt to the number of crashes they estimate in their group, by communicating over denser expander graphs from a family of $\Theta(\log{n})$ expanders: $\cG_{\texttt{in}} = \{G_{\texttt{in}}(0), \ldots,  G_{\texttt{in}}(\log{n}),\}$, every time they observe a significant reduction of non-faulty processes in their neighborhood. 

Initially, processes from $\cP_{1}$ span an expander graph with $O(\log{n})$ degree on the set $\cP_{1}$, denoted $G_{\texttt{in}}(0)$. In the course of execution each process from $\cP_{1}$ will test the number of non-faulty processes in the $O(\log{n})$ neighborhood in $G_{\texttt{in}}(0)$. If the number appeared to be too small, the process will upgrade the expanding graph it uses by doubling its degree, namely it switches to the next graph from the family - $G_{\texttt{in}}(1)$. From now on, this process will use this denser graph for the testing. The ultimate goal of this 'densification' is to enable each process communication with a constant fraction of alived other process from $\cP_{1}$. Note here, that this process of adaptive adjustment to failures pattern happens independently for processes in $\cP_{1}$.  

The communication of processes from $\cP_{1}$ with processes from $\cP_{2}$ is strictly correlated with their estimation of the number of processes being alive in their $O(\log{n})$ neighborhood in part $\cP_{1}$. Initially, a process from $\cP_{1}$ sends its rumor according to other expander graph $G_{\texttt{out}}^{0}$ of degree $O(\log{n})$, the first graph from family of expanders graphs $\cG_{\texttt{out}} = \{G_{\texttt{out}}(0), \ldots,  G_{\texttt{out}}(\log{n})\}$ and each time the process chooses a denser graph from family $\cG_{\texttt{in}}$ it also switches to a denser graph from family $\cG_{\texttt{out}}$. The intuition is that if a process knows that its number of neighbours in $O(\log{n})$ neighborhood has been reduced by a constant factor since it checked it last time, it can afford an increase of its degree in communication with $\cP_{2}$ by the same constant factor, as the amortized message complexity should stay the same.

\Paragraph{Estimating the number of alive processes in $O(\log{n})$ neighborhoods.} In the heart of the above method lies an algorithm, called \textsc{LocalSignaling} that for each process $p$, tests the number of other alive processes in $p$'s neighborhood of radius $O(\log{n})$. As a side result, it also allows to exchange a message with these neighbors. The algorithm takes as in input: a set of all processes in the system $\cP$, an expander-like graph family $\cG = \{G(0), \ldots, G_{t}\}$ spanned on $\cP$, together with two parameters $\delta$ and $\gamma$, describing a diameter and a maximal degree of the base graph $G(0)$; the name of a process $p$; the process' level $\ell$ which denotes which graph from family $\cG$ the process uses to communicate; and the message to convey $r$. Let $\cT$ denote a graph $\cup_{v \in \cP} N_{G_{\ell_{v}}}(v)$, that is a graph with set of vertices corresponding to $\cP$ and set of edges determined based on neighbors of each vertex from a graph on the proper level. Provided that \textsc{LocalSignaling} is executed synchronously on the whole system it returns whether the process $p$ was connected to a constant number of other alived processes at the beginning of the execution accordingly to graph $\cT$. Assumed that, the algorithm guarantees that $p$'s message reached all these processes and vice versa - messages of these processes reached $p$. On the other hand, we will prove that the amortized bit complexity of a synchronous run of the $\textsc{LocalSignaling}$ algorithm is $\logO(n)$. This is the most advanced technical part used in our algorithm. It's full description and detailed analysis is given in Section~\ref{sec:local-signaling-short}.

\Paragraph{\textsc{BipartiteGossip} algorithm and its analysis.}
In this paragraph we give a pseudocode of the \textsc{BipartiteGossip} algorithm which implements the idea discussed before. We start by formal description of utilized graphs and connected to them subroutines.

The graphs used by processs are grouped into two families: $\cG_{\texttt{in}}$ and $\cG_{\texttt{out}}$. Denote $t = \floor{\log{n}}$, $\delta = 2 \log{n}$, $\gamma = 24 \log{n}$. Consider a process $p$; it gets as an input the partition of set $[n]$ into groups $P_{1}$, $P_{2}$, hence it can determine the group it belongs to. %
The family $\cG_{\texttt{in}} = \{G_{\texttt{in}}(0), \ldots, G_{\texttt{in}}(t + 1)\}$ serves for communication {\em inside} each group. 

A single graph $G_{\texttt{in}}(i)$, for $i \in \{0, \ldots, t\}$, is a union of $G(n/2, \frac{n}{3\cdot2^{j}}, \delta, \gamma)$, over $j \in \{0, \ldots, i\}$, of graphs given in the Theorem~\ref{theorem:overlay-graphs-exist} with nodes being the processes in $p'$s group, that is $G_{\texttt{in}}(i) = \bigcup_{j = 0}^{j = i} G(n/2, \frac{n}{3\cdot2^{j}}, \delta, \gamma)$.
Graph $G_{t+1}$ is a clique with nodes being the processes of $p$'s~group.

The family $\cG_{\texttt{out}} = \{G_{\texttt{out}}(0), \ldots, G_{\texttt{out}}(t + 1)\}$ serves for communication {\em outside} each group.
A single graph $G_{\texttt{out}}(i)$, for $i \in \{0, \ldots, t\}$, is a union of $G(n, \frac{2n}{3\cdot2^{j}}, \delta, \gamma)$, over $j \in \{0, \ldots, i\}$, of graphs given in the Theorem~\ref{theorem:overlay-graphs-exist} with nodes being all the processes, that is $G_{\texttt{out}}(i) = \bigcup_{j = 0}^{j = i} G(n, \frac{2n}{3\cdot2^{j}}, \delta, \gamma)$.
Graph $G_{t+1}$ is a clique with nodes being all the processes.

Observe, that those families and parameters $t, \delta, \gamma$ are deterministic and can be precomputed by each process, assumed the knowledge of partition $P_{1}$ and $P_{2}$. As a such, they are assumed to be known to the algorithm on every stage of the algorithm.

\textbf{The} \texttt{Exchange} \textbf{communication scheme for a graph $G$, used in the \textsc{BipartiteGossip} algorithm:} 
This communication scheme takes two rounds. In the first round $p$ sends a message containing a bit and the set $R$, being a set of all learned so far rumors by $p$, to every process in the set $N_{G}(p)$ that is not faulty according to $p$'s view on the system. The receiver treats such a message as both a request and an increment-knowledge message. In the second round, $p$ responses to all the received requests by sending $R$ to each sender of every request received in the previous round.

\begin{algorithm}[H]
\SetAlgoLined
\SetKwInput{Input}{input}
\SetKwInput{Output}{output}
\SetKwFor{RepTimes}{repeat}{times}{}
\Input{partition $\cP_{1}$, $\cP_{2}$; $p$, $r$}
\Output{set $R$ of learned rumors, initially set to $\{r\}$}
\For{$i \leftarrow 1$ \KwTo $2t$}
{
    \RepTimes{3} 
    {
        \label{alg:bp-gossip:second-loop} 
        do \textsc{Exchange} on graph $\cG_{\texttt{out}}(i + 1)$\; \label{alg:bp-gosip:external-com}
        \RepTimes{$2\gamma + 1$}
        {
            \label{alg:bp-gossip:internal-com-loop}
            do \textsc{Exchange} on graph $\cG_{\texttt{in}}(i + 7)$\;
        }
        \RepTimes{t + 2}
        {
            do \textsc{Exchange} on graph $\cG_{\texttt{in}}(i + 2)$\; \label{alg:bp-gossip:internal-local}
            $\texttt{survived} \leftarrow \textsc{LocalSignaling}(p, \cG_{\texttt{in}}, i, \delta, \gamma, R)$\;
            \If{$\texttt{survived} = \texttt{false}$}
            {
                $i \leftarrow \min(i + 1, t + 1)$
            }
        }
    } 
}
\Return{$R$}
\caption{\textsc{BipartiteGossip}}
\end{algorithm}

\Paragraph{Analysis of correctness.}

We call a single iteration of the main loop of the $\textsc{BipartiteGossip}$ algorithm an \textit{epoch}.
First, we show that if in a single epoch a big fraction of processes from the groups $P_{1}$ and $P_{2}$ worked correctly, then 
by the end of the epoch every process has learned both rumors $r_{1}$ and $r_{2}$.
Let $\cE$ be an epoch.
Let $\texttt{BEGIN}_{1}$ ($\texttt{BEGIN}_{2}$) be the set of processes from the group $P_{1}$ (group $P_{2}$ respectively) that were non-faulty before the epoch $\cE$ started. Let $\texttt{END}_{1}$ ($\texttt{END}_{2})$ be the set of those processes from the group $P_{1}$ (group $P_{2}$ respectively) that were non-faulty after the epoch $\cE$ ended. We assume that epoch $\cE$ is such that: 
\begin{quote}
	$|\texttt{END}_{1}| > \frac{1}{3}|\texttt{BEGIN}_{1}|$ and $|\texttt{END}_{2}| > \frac{1}{3}|\texttt{BEGIN}_{2}|$.
\end{quote}

\begin{lemma} \label{lemma:level-lower-bound}
After the first iteration of the loop from line~\ref{alg:bp-gossip:second-loop} in epoch $\cE$, each non-faulty process from the group $P_{1}$ is on level $j_{\p} \ge \log\big(\frac{n}{3\cdot64\cdot|\texttt{BEGIN}_{1}|}\big)$.
\end{lemma}
\begin{proof}
Assume, to the contrary, that there is a process $p \in P_1$ being on level $j_{\p}$ strictly smaller then $\log\big(\frac{n}{3\cdot64\cdot|\texttt{BEGIN}_{1}|}\big)$ at the end of phase $1$ of epoch $\cE$. Since in each iteration of the loop~\ref{alg:bp-gossip:second-loop}, an instance of $\textsc{LocalSignaling}$ is executed $t+2 = |\cG_{\texttt{in}}| + 2$ times, 
process $p$ must have survived at least one $\textsc{LocalSignaling}$ execution while being at that or a smaller level. 
In this execution, process $p$ was using graph $G_{\texttt{in}}(j_{p})$ that satisfies the conditions of Theorem~\ref{theorem:overlay-graphs-exist} with parameter $k_{j_{p}} := \frac{n}{3\cdot2^{j_{p}}}$. From the specific properties of the \textsc{LocalSignaling} algorith, i.e. Lemma~\ref{lemma:probing} in Section~\ref{sec:local-signaling-short}
point $1$, we conclude that $p$ must have a $(\gamma, \delta)$-dense-neighborhood in $G_{j_{p}} \cap \texttt{BEGIN}_{1}$. A property of the overlay graphs, Lemma~\ref{lemma:sparse-to-dense}, says that any $(\gamma, \delta)$-dense-neighborhood in the graph $G_{j_{p}}$ has at least $\frac{n}{64\cdot3\cdot2^{j_{p}}}$ nodes. Given that $j_{p} < \log\big(\frac{n}{3\cdot64\cdot|\texttt{BEGIN}_{1}|}\big)$, we conclude that the size of the $(\gamma, \delta)$-dense-neighborhood of $p$ in $G_{j_{p}} \cap \texttt{BEGIN}_{1}$ is at least $\frac{n}{64\cdot3\cdot2^{j_{p}}} > |\texttt{BEGIN}_{1}|$. This gives a contradiction with the fact that the set $\texttt{BEGIN}_{1}$ contains all non-faulty process from the group $P_{1}$.
\end{proof}

\begin{lemma} \label{lemma:fraction-know}
There exists a set $C_{1} \subseteq \texttt{END}_{1}$ of size at least $\frac{|\texttt{BEGIN}_{1}|}{4}$ such that after the second iteration of the loop~\ref{alg:bp-gossip:second-loop} of epoch $\cE$ each process $p$ from set $C_{1}$ has the other rumor $r_{2}$ in its set $\cR$.
\end{lemma}
\begin{proof}
Let $i = \ceil{\log{\frac{n}{3\cdot64\cdot|\texttt{BEGIN}_{1}|}}}$. From Lemma~\ref{lemma:level-lower-bound} we know that from in the beginning of the second iteration of the loop~\ref{alg:bp-gossip:second-loop} of epoch $\cE$ each process is at level at least $i$. Therefore, starting from the second iteration of this loop, each process uses graph $G_{\texttt{in}}(i + 7)$ (or a denser graph in the family $\cG_{\texttt{in}}$) to communicate within processes from the same group.  The set $\texttt{END}_{1}$, viewed as a set of nodes in the graph $G_{\texttt{in}}(i + 7)$, is of size at least $\frac{|\texttt{BEGIN}_{1}|}{3}$. 
Now, we constructed graph $G_{\texttt{in}}(i + 7)$ such that it is $(k_{i + 7}, 3/4, \delta)$-compact, where  $k_{i + 7} := \frac{n}{3\cdot2^{i + 7}}$. Because $i \ge \log{\frac{n}{3\cdot64\cdot|\texttt{BEGIN}_{1}|}}$, thus $k_{i+7} < \frac{|\texttt{BEGIN}_{1}|}{3} < |\texttt{END}_{1}|$. Therefore, there exists a survival set $C_{1}$ in graph $G_{\texttt{in}}(i + 7)$ being a subset of $\texttt{END}_{1}$. The size of $C_{1}$ is at least $|\texttt{END}_{1}| \cdot  3/4 > \frac{|\texttt{BEGIN}_{1}|}{4}$.

Analogical reasoning proves that after the first iteration of the loop~\ref{alg:bp-gossip:second-loop} of epoch $\cE$ each process from set $P_{2}$ is on level $j \ge \ceil{\log{\frac{n}{3\cdot64\cdot|\texttt{BEGIN}_{2}|}}}$ and there exists a set $C_{2} \subseteq \texttt{END}_{2}$, such that $|C_{2}| > \frac{|\texttt{BEGIN}_{2}|}{4}$.
Without loss of generality assume that $j \ge i$ (the communication between non-faulty processes is bi-directional). 
The processes from set $C_{2}$ use overlay graph $G_{\texttt{out}}(j + 1)$ to communicate with the group $\cP_{1}$ in the beginning of the second iteration of the loop~\ref{alg:bp-gossip:second-loop} in epoch $\cE$, i.e. to execute line~\ref{alg:bp-gosip:external-com}. This graph is $(\frac{n}{3\cdot64\cdot2^{j}})$-expanding. Since $j \ge i$, both sets $C_{1}$ and $C_{2}$ are of size at least $\frac{n}{3\cdot64\cdot2^{j}}$. Hence, due to the graph expansion and proper sizes of $C_1,C_2$, there exists an edge between $C_{1}$ and $C_{2}$ in graph $G_{\texttt{out}}(j + 1)$. Thus, the call of the \textsc{Exchange} communication schemes on graph $G_{\texttt{out}}(j + 1)$, that takes place in line~\ref{alg:bp-gosip:external-com} of epoch $\cE$, results in at least one process from set $C_{1}$ knowing the other rumor $r_{2}$. 

From another property of the overlay graphs, Lemma~\ref{lemma:survival-small-diameter}, we know that every other pair of processes in $C_{1}$ are in distance $2\gamma + 1$ in  graph $G_{\texttt{in}}(i+7)_{|C_{1}}$. Therefore, after the execution of the loop in line~\ref{alg:bp-gossip:internal-com-loop} in the second iteration of the loop~\ref{alg:bp-gossip:second-loop} in epoch $\cE$, each process from $C_{1}$ knows the other rumor $r_{2}$.
\end{proof}

\begin{lemma}
\label{lemma:non-crashed-has-rumor}
After the epoch $\cE$ ends, \textit{each} process from the set $\texttt{END}_{1}$ knows the other rumor $r_{2}$.
\end{lemma}
\begin{proof}
Consider any process $p$ from the set $\texttt{END}_{1}$. In the third iteration of the loop~\ref{alg:bp-gossip:second-loop} in the epoch $\cE$, there exists at least one round in which that process survives the procedure $\textsc{LocalSingaling}$. Assume that $p$ survives that instance of the Local Signaling with the level set to $i_{p}$. A property of the \textsc{LocalSignaling} algorithm, i.e. Lemma~\ref{lemma:probing} point $2$, from Section~\ref{sec:local-signaling-short}, gives us that there exists a $(\gamma, \delta)$-dense-neighborhood of process $p$ in the graph $G_{\texttt{in}}(i_{p})$ consisting of processes that are non-faulty and at the level at least $i_{p}$ at the beginning of the third iteration of the loop~\ref{alg:bp-gossip:second-loop}. Moreover, the $(\gamma, \delta)$-dense-neighborhood is such that $p$ received the set $\cR$ of any node from the set in this instance of the \textsc{LocalSignaling} algorithm. 
Let $D$ be the set of those processes that constitute the $(\gamma, \delta)$-dense-neighborhood. From Lemma~\ref{lemma:sparse-to-dense}
we know that the size of $D$ is at least $\frac{n}{64\cdot3\cdot2^{i_{p}}}$. 

The graph $G_{\texttt{in}}(i_{p} + 2)$ used to in the communication rounds that precedes that instance (i.e. to execute line~\ref{alg:bp-gossip:internal-local}) of Local Signaling is $\big(\frac{n}{64\cdot3\cdot 4\cdot2^{i_{p}}}\big)$-expanding. 
Consider the set $C_{1}$ given in Lemma~\ref{lemma:fraction-know}. We have $|C_{1}| \ge \frac{|BEGIN_{1}|}{4}$.
We argue that set $C_{1}$ has size at least $\frac{n}{64\cdot3\cdot 4\cdot 2^{i_{p}}}$. This holds because Lemma~\ref{lemma:level-lower-bound} bounds the value $i_{p}$ from below by $\log\big(\frac{n}{3\cdot64\cdot|\texttt{BEGIN}_{1}|}\big)$.
Therefore, by expansion of the graph, the sets $D$ and $C_1$ are connected by at least one edge in $G_{\texttt{in}}(i_{p} + 2)$.
From Lemma~\ref{lemma:fraction-know} we derive that each process in $C_{1}$ knows the other rumor $r_{2}$ at the beginning the third iteration of the loop~\ref{alg:bp-gossip:second-loop} in epoch $\cE$. 

Hence, the rumor $r_{2}$ must have reached some process in $D$ before the instance of Local Signaling started, when processes from $D$ were performing the communication inside their group. This holds, because each process in $D$ used the graph $G_{\texttt{in}}(i_{p}+2)$, or a denser graph from the family $\cG_{\texttt{in}}$ (which, by definition, has graph $G_{\texttt{in}}(i_{p}+2)$ as a subgraph)
, in the rounds preceding the Local Signaling, i.e. line~\ref{alg:bp-gossip:internal-local}.
Next, in the execution of Local Signaling which $p$ survived, the information from any process from set $D$ was conducted to process $p$, and this information includes the other rumor $r_{2}$.
\end{proof}

\Paragraph{Analysis of communication complexity.}
Let $L_{i}(r)$ be the set of non-faulty processes that at the beginning of the round $r$ are on level $i$ or bigger. 
We show that for any round $r \ge 2$ and for any $i \in [t]$, the number $|L_{i}(r)|$ is at most $\frac{2n}{2^{i}}$.

\begin{lemma}
\label{lemma:dense-are-small}
For any round $r \ge 2$ and any level $i \in [t]$ the number of processes in the set $L_{i}(r)$ is at most $\frac{2n}{2^{i}}$.
\end{lemma}
\begin{proof}
Assume, to arrive at a contradiction, that there exists round $r \ge 2$ and level $i \in [t]$ such that at the beginning of round $r$ the inequality $|L_{i}(r)| > \frac{2n}{2^{i}}$ holds. Consider graph $G_{\texttt{in}}(i - 1)$. The construction of the graph guarantees that it is $(\frac{n}{3\cdot2^{i - 1}}, 3/4, \delta)$-compact. Because $\frac{2n}{2^{i}} \ge \frac{n}{3\cdot2^{i - 1}}$, thus there exists a survival subset $S$ of $L_{i}(r) \cap G_{\texttt{in}}(i - 1)$ of size at least $\frac{3\cdot n}{2^{i - 1}} > 0$, because $i \le t = \floor{\log{n}}$. Let $r'$ be the maximum round in which an instance of the \textsc{LocalSignaling} algorithm started and there exists a process from set $S$ that executed the \textsc{LocalSignaling} at level \textit{exactly} $i-1$. Let $A \subseteq S$ be the set of the processes at level exactly $i-1$ in round $r'$.

First, since all processes from $S$ are at level at least $i$ in round $r$, round $r'$ exists, furthermore $r' < r$ and the set $A$ is non-empty. Also, every process from $S$ starts the instance of the \textsc{LocalSignaling} in round $r'$ at level $i-1$ or bigger. The last is true, because not surviving an instance of Local Signaling by a process results in increasing its level by $1$.

Observe, that all processes used the graph $G_{\texttt{in}}(i - 1)$ as a subgraph of the communication graph in the instance of the \textsc{LocalSignaling} algorithm starting at round $r'$. Since set $S$ is a survival set for $L_{i}(r)$ of the graph $G_{\texttt{in}}(i-1)$, thus, from a property of the \textsc{LocalSignaling} algorithm, that is Lemma~\ref{lemma:probing} point $3$ in Section~\ref{sec:local-signaling-short}, we conclude that every process from set $S$ that started this instance of Local Signaling at level $i-1$ survived this instance of Local Signaling. In particular this means that processes from the non-empty set $A \subseteq S$ stayed at level exactly $i-1$ after this instance of Local Signaling. This contradicts the fact that $r'$ was defined as the last round with a process in $S$ starting Local Signaling at level exactly $i-1$.
\end{proof}

\noindent Let us recall the Theorem~\ref{thm:cheap-bi-gossip}.
\begin{reptheorem}{thm:cheap-bi-gossip}
\textsc{BipartiteGossip} solves the Bipartite Gossip problem in $O(\log^3 n)$ rounds, with $O(\log^5 n \cdot |\cR|)$~amor\-tized number of communication bits, where $|\cR|$ is the number of bits needed to encode the~rumors.
\end{reptheorem}

\begin{proof}
In order to count the number of bits sent, in total, by all processes, observe that a process that is at a level $i$ in a round uses at most $O(\frac{26n\cdot \delta}{k_{i}})$, where $k_{i} := \frac{n}{3\cdot 2^i}$, links to communicate in this round with other processes. Lemma~\ref{lemma:dense-are-small} assures that there is at most $\frac{2n}{2^{i}}$ processes at level $i$ in a round. Thus, in a single round, processes use $O( \sum_{i=0}^{i=t+1} \frac{2n}{2^{i}} \cdot \frac{2 \cdot 26n\cdot \delta}{k_{i}}) = O(t \cdot n \cdot \delta) = O(n \cdot \log^{2}{n})$ messages. A single message carries a single bit and at most two rumors. The number of bits needed to deliver such essage is $O(|\cR|)$. Since the algorithm runs in $O(\log^3 n)$ rounds, the total number of bits used by processes is  $O(n \cdot \log^5 n \cdot |\cR|)$.

In order to prove correctness, observe that if all the processes from group $P_{1}$ or all the processes from group $P_{2}$ fail during the execution, then every non-failed process knows the rumor of every other non-failed process (because only the owners of a single rumor survived).
Hence, consider a case in which at the end of the algorithm the number of processes that survived is greater than zero in each group. 
Since the number of epochs is $2t$, there must exist an epoch $\cE$ in which the ratio of the processes that survived the epoch to the processes that were non-faulty at the begin of the epoch is greater then $\frac{1}{3}$, in both groups. In every epoch in which the above is not satisfied, the number of non-faulty processes decreases at least 3 times in one of the groups, and thus it cannot happen more then $2\log n < 2t$ times. The conclusion of Lemma~\ref{lemma:non-crashed-has-rumor} completes the proof of correctness. 
\end{proof}

\subsection{The \textsc{Gossip} algorithm}
\label{sec:gossip-long}

Here, we describe an algorithm based on the divide-and-conquer approach, called \textsc{Gossip} that utilizes the \textsc{BipartiteGossip} algorithm to solve Fault-tolerant Gossip.
Each process takes the set $\cP$, an initial rumor $r$ and its unique name $p \in [|\cP|]$ as an input.
The processes split themselves into two groups of size at most $\ceil{n/2}$. The groups are determined based on the unique names. The first $\ceil{n/2}$ processes with the smallest names make the group $\cP_{1}$, while the $n - \ceil{n/2}$ processes with the largest names define the group $\cP_{2}$. Each of those two groups of processes solves Gossip separately by evoking the \textsc{Gossip} algorithm inside the group only. The processes from each group know the names of every other process in that group, hence the necessary conditions to execute the \textsc{Gossip} recursively are satisfied. After the recursion finishes, a process from $\cP_{1}$ stores a set of rumors $\cR_{1}$ of processes from its group, and respectively a process from $\cP_{2}$ stores a set of rumors $\cR_{2}$ of processes from its group. Then, the processes solve Bipartite Gossip problem by executing the \textsc{BipartiteGossip} algorithm on the partition $\cP_{1}$, $\cP_{2}$ and having initial rumors $\cR_{1}$ and $\cR_{2}$. The output to this algorithm is the final output of the \textsc{Gossip}.

\remove{
A standard induction analysis of recursion and Theorem~\ref{thm:cheap-bi-gossip} led to the following theorem, which proof is deferred to Section~\ref{sec:gossip-long}.
}

\begin{reptheorem}{thm:cheap2-gossip}
\Gossip{} solves deterministically the Fault-tolerant Gossip problem in $O(\log^3 n)$ rounds using $O(\log^6 n \cdot |\cR|)$~amortized number of communication bits, where $|\cR|$ is the number of bits needed to encode the~rumors.
\end{reptheorem}

\begin{proof}
Because of the recursive nature of the algorithm, the easiest way to analyze it is by using the induction principle over the number of processes. If the system consists of one non-faulty process, the process returns the exact number of zeros and ones immediately, regardless of its initial bit. Thus, both the conditions -- termination and validity -- are satisfied.

Assume then, that the system consists of $n > 1$ processes. First, the processes perform the \textsc{Gossip} algorithm in two groups of size at most $\ceil{n/2}$. It takes $T(\ceil{n/2})$ rounds and $2M(\ceil{n/2})$ bits, where $T(x)$ and $M(x)$ is the number of rounds and the total number of bits used by the \textsc{Gossip} algorithm executed on a system with $x$ processes. Then, the $n$ processes execute the \textsc{BipartiteGossip} algorithm, which requires $O(\log^{3}{n})$ rounds and $O(n\cdot\log^{6}{n})|\cR|$ communication bits, by Theorem~\ref{thm:cheap2-gossip}. Thus, in total, the algorithm takes $T(\ceil{n/2})+O(\log^{3}{n})$ rounds and sends $M(\ceil{n/2} + n\cdot\log^{6}{n}|\cR|)$ communication bits. Given that $T(1) = 1$ and $M(0) = 0$, we calculate that the functions $T(x)$ and $M(x)$ are asymptotically equal to $O(\log^{3}{n})$ and $O(n\cdot\log^{7}{n})|\cR|$, respectively. This proves the termination condition and bounds the use of communication bits.

Now, we prove that the validity condition holds. After the recursive run of the algorithm, each process from the group $P_{1}$ stores set $\cR_{1}$ consisting of rumors of alived processes from $\cP_{1}$. This set satisfies the validity conditions for the system consisting of processes from the group $P_{1}$. The processes from the group $P_{2}$ store analogical set $\cR_{2}$. If all processes from the group $P_{1}$ or $P_{2}$ have crashed, then the validity condition holds from the inductive assumption. If there exists at least one correct process in each group, then the execution of the \textsc{BipartiteGossip} algorithm guarantees that each process has sets $R_{1}$ and $R_{2}$. In this case, the result returned by every process, that is, the union of these two sets, satisfies the validity condition.
\end{proof}

\remove{Now, we prove that the validity condition holds. After the recursive run of the algorithm, each process from the group $P_{1}$ stores two values -- $\texttt{zeros}$  and $\texttt{ones}$ -- that estimate the number of processes in $P_{1}$ starting with the initial bit $0$ and $1$, respectively. Those values satisfy the validity conditions for the system consisting of processes from the group $P_{1}$. The processes from the group $P_{2}$ store analogical values but correspond to the number of initial bits in the group $P_{2}$. If all processes from the group $P_{1}$ or $P_{2}$ have crashed, then the validity condition holds from the inductive assumption. If there exists at least one correct process in each group, then the execution of the \textsc{Cheap2-Gossip} algorithm guarantees that each process has both pairs of the values. That is, the estimated number of initial bits equal to zero and one, from both groups $P_{1}$ and $P_2$. In this case, the result returned by every process, that is, the sum of the number of zeros from each pair and the number of ones from each pair satisfies the validity condition.
}

\Paragraph{Modification for Fuzzy Counting.}
We define the \textit{Fuzzy Counting} problem as follows. There is a set $n$ processes, $\cP$, with unique names that are comparable. Each process knows the names of other processes (i.e. they operate in KT-$1$ model). 
Each process starts with an initial bit $b \in \{0, 1\}$. Let $\texttt{Z}$ denote the number of processes that started with the initial bit set to $0$ and never failed. Similarly, $\texttt{O}$ denotes the number of processes that started with $1$ and never failed.
Each process has to return two numbers: $\texttt{zeros}$ and $\texttt{ones}$. 
An algorithm is said to solve \blurred counting if every non faulty process terminates (termination condition) and the values returned by any process fulfill the 
conditions: $\texttt{zeros} \ge |\texttt{Z}|$, $\texttt{ones} \ge |\texttt{O}|$ and $\texttt{zeros} + \texttt{ones} \le n$ (validity~condition).

To solve this problem, we use the \textsc{Gossip} algorithm with the only modification that now we require the algorithm the return the values $\texttt{Z}$ and $\texttt{O}$, instead of the set of learned rumors. We apply the same divide-and-conquer approach. That is, we partition $\cP$ into groups $\cP_{1}$ and $\cP_{2}$ and we solve the problem within processors of this partition.  Let $\texttt{Z}_{1}$, $\texttt{O}_{1}$ and $\texttt{Z}_{2}$, $\texttt{O}_{2}$ be the values returned by recursive calls on set of processes $\cP_{1}$ and $\cP_{2}$, respectively. Then, we use the \textsc{BipartiteGossip} algorithm to make each process learn values $\texttt{Z}$ and $\texttt{O}$ of the other group. Eventually, a process returns a pair of values $\texttt{Z}_{1} + \texttt{Z}_{2}$ and $\texttt{O}_{1} + \texttt{O}_{2}$ if it received the values from the other partition during the execution of \textsc{BipartiteGossip}; or it returns the values corresponding to the recursive call in its partition otherwise. 
It is easy to observe, that during this modified execution processes must carry messages that are able to encode values $\texttt{Z}$ and $\texttt{O}$, thus in this have it holds that $|\cR| = O(\log{n})$. We conclude this modification in the following theorem.

\begin{reptheorem}{thm:fuzzy-counting}
There exists an algorithm, called \textsc{FuzzyCounting} that solves Fuzzy Counting problem in $O(\log^{3}{n})$ rounds with $O(\log^{7}{n})$ amortized bit complexity.
\end{reptheorem}

\section{Local Signalling -- Estimating neighborhoods in expanders} \label{sec:local-signaling-short}

The \textsc{LocalSignaling} algorithm, presented in this section, allows to adapt the density of used overlay graph to any malicious fail pattern guaranteeing fast information exchange among a constant fraction of non-faulty nodes with amortized $\logO(n |\cR|)$ bit complexity, where $\cR$ is the overhead that comes from the bit size of the information needed to convey.

\Paragraph{High level idea.}
The procedure is formally denoted $\texttt{LocalSignaling}(\cP, p, \cG, \delta, \gamma, \ell, r)$, where $\cP$ is the set of all processes, $p$ is the process that executes the procedure and $\cG = \{G(1), \ldots, G(t)\}$ denotes the family of overlay graphs that processes from $\cP$ uses to select processes to directly communicate -- those are neighborhoods in some graph of the family $\cG$.  In our case, the family will consist of graphs with increasing connectivity properties. Parameters $\gamma,\delta$ correspond to the property of $(\gamma,\delta)$-dense-neighborhoods which the base graph $G(1)$ must fulfill. They are also related to the time and actions taken by processes if failures occur, respectively.
The parameter $\ell \le t$ is called a \emph{starting level} of process $p$ and denotes the communication graph from family $\cG$ from which the node $p$ starts the current run of the procedure. This parameter may be different for different processes.
Finally, the parameter $r$ denotes a rumor that process $p$ is supposed to deliver to other processes.
Since processes operates in KT$-1$ model, the implementation assumes that each process uses the same family $\cG$ (see the corresponding discussion after Theorem~\ref{theorem:overlay-graphs-exist}).

\remove{
The detailed pseudocode of the procedure 
is given in Section~\ref{sec:local-signaling-long}, Figure~\ref{fig:local-signaling}, while
below we present its high-level~description.
}

\Paragraph{Procedure~$\textsc{LocalSignaling}(\cP, p, \cG, \delta, \gamma, \ell, r)$} takes $2\gamma$ consecutive rounds. The level of process $p$ executing the procedure is initially set to $\ell$, and is stored in a local variable $i$. Each process stores also s set $R$ of all rumors it has learned to this point of execution. Initially, $R$ is set to $\{r\}$.
\BB
\begin{description}
\item[Odd rounds:]
Process~$p$ sends a request message to each process~$q$ in~$N_{G(i)}(p)$,
provided $i>0$.
\remove{
Based on the number of received messages, the level $i$ of $\p$ may decrease in the beginning of the next round to level $i_{1} \ge 0$, indicating that for $j$ such that $i_{1} < j \le \ell$ graph $G_{j}$ was too sparse to allow $\p$ to communicate with a fraction of all non-faulty processes. 
}
\BB
\item[Even rounds:]
Every non-faulty process $q$ responds to the requests received at the end of the previous round -- by replying to the originator of each request a message containing the current level $i$ of process $q$ and the set $R$ of all different rumors $q$ collected so far. 

At the end of each even round, processes that 
requested information in the previous round collect the responses to those requests. 
If a single process $p$ received less then $\delta$ responses with level's value of its neighbors greater or equal than its level value $i$, then $p$ decreases $i$ by one. Additionally, $p$ merges every set of rumors it received with its own set $R$.
If $i$ drops to $0$, then $p$ does not send any requests in the consecutive rounds. 
\BB
\item[Output:]
We say that process $p$ \emph{has not survived} the \textsc{LocalSignaling} algorithm if it ends with value $i$ lower than its initial level $i$. Otherwise, $p$ is said to have \emph{survived} the \textsc{LocalSignaling} algorithm. $p$ returns a single bit indicating whether it has survived or not and the set $R$ containing all rumors it has learnt in the course of the execution.
\end{description}

\begin{lemma}
\label{lem:signaling-complexity}
The procedure $\textsc{LocalSignaling}(\cP, p, \cG, \delta, \gamma, \ell, r)$ takes $O(\gamma)$ rounds and uses $\big( \sum_{i = 1}^{i = t} |L_{i}|\cdot|N_{G_{\le i}}(L_{i})| \cdot \gamma \cdot |\cR| \big)$ communication bits, where $L_{i}$ denotes the set of processes that start at level $i$, the graph $G_{ \le i}$ is a union of graphs $G(1), \ldots, G(i)$, and the value $|\cR|$ denotes the number of bits needed to encode all possible rumors.
\end{lemma}

\begin{proof}
Each process executes work for $2\gamma$ rounds, thus this must be also the running time of the whole procedure. 
Next, we bound the total number of bits that processes used in the instance of Local Signaling. 
Observe, that every message is of size at most $1 + |\cR|$, thus it is enough to upper bound the total number of messages sent. Each node in each round either sends a request or replies once to each received requests. Thus, it is enough to bound the number of sent requests only. The processes that start at level $i$ may only decrease their levels after a round. There are $O(\gamma)$ rounds in total, thus they send at most $|N_{G_{\le i}}(L_{i})| \cdot \gamma$ requests. 
If we sum this expression over all possible start levels, we get the claimed upper bound on the number of messages and, in consequence, the claimed upper bound on the number of bits used by participating processes.
\end{proof}


\Paragraph{Surviving 
the \textsc{LocalSignaling} algorithm -- the consequences.}
Here, we present benefits of the \textsc{LocalSignaling} algorithm if a proper graph family $\cG$ is used.
Assume that $t \ge 1$ and consider a sequence $(k_{i})_{i \in [t]}$.
Let $\cG = \{G(1), \ldots, G(t)\}$ be a family of graphs $G(i)=G(n, k_{i}, \delta, \gamma)$ defined as in Theorem~\ref{theorem:overlay-graphs-exist}. We require, for any $1 \le i < t$ that $G(i) \subseteq G(i+1)$. 
Consider a simultaneous run of the procedure $\textsc{LocalSignaling}(\cP, p, \cG, \delta, \gamma, \ell, r)$ at every process~$p \in \cP$. Here, we require each process $p \in \cP$ to use the same family of graphs $\cG$. Since our processes operates in KT$-1$ model, this requirement could be always satisfied.

When a process~$p$ survives an instance of Local Signaling and no failures occurred in this instance, 
then the set of processes which exchanged a message with~$p$ during the Local Signaling execution is a $(\gamma,\delta)$-dense-neighborhood for~$p$ in graph $G_{\ell}$, where $\ell$ is the level of $p$ given as an argument to the procedure. We show that the same property holds in general case (when failures occurred), provided $p$ has survived Local Signaling. 

Let $B_{\ell,1}$ be \emph{the start set on level $\ell$}: it consists of the processes that are non-faulty at the beginning of this instance of Local Signaling and their level is at least $\ell$.
Let $B_{\ell,2}\subseteq B_{\ell,1}$ be \emph{the end set}: it consists of the processes that are non-faulty just after the termination of this instance 
and their level at the beginning of this instance was at least $\ell$.
The processes in~$B_{\ell,1}\setminus B_{\ell,2}$ are among those that have crashed during the considered instance of Local Signaling.

\BB
\begin{lemma}
\label{lemma:probing}
The following properties hold for arbitrary times of crashes of the processes in~$B_{\ell,1}\setminus B_{\ell,2}$:

1. If there is a $(\gamma,\delta)$-dense-neighborhood for $p \in B_{\ell,2}$ in graph~$G_{\ell}|_{B_{\ell,2}}$, 
then process $p$ survives Local Signaling. 

2. If $p$ survived the Local Signaling, then there is $(\gamma,\delta)$-dense-neighborhood for~$p \in {B_{\ell,1}}$ in graph~$G(\ell)|_{B_{\ell,1}}$. Moreover, $p$ receives the rumor $r$ of any node from that $(\gamma,\delta)$-dense-neighborhood. 

3. Any process in a survival set $C$ for $B_{\ell,2}$ that started at 
level exactly $\ell$ survives Local~Signaling.
\end{lemma}

\begin{proof}
We first prove property~$1$.
Let $S$ be any $(\gamma,\delta)$-dense-neighborhood for $p$ in graph~$G|_{B_{\ell,2}}$. 
We argue, that every process in~$S \cap N^{\gamma-1}_{G(\ell)}(p)$ receives at least $\delta$ responses at the second round of Local Signaling. Indeed, at least that many requests were sent by this process to all its neighbors, in the first round. Since, all these neighbors were at the level at least $\ell$ in the first round (by the definition of~${B_{i_{\p},2}}$), 
thus the process preserves its variable $i$ set to $\ell$ after the second round of the Local Signaling procedure. 
By induction on $j \le \gamma$, no process in~$S\cap N^{\gamma-j}_{G(\ell)}(p)$ decreases its value $i$ before the end of the $2j$-th round of Local Signaling, and hence $p$ survives.

Next, we argue that property~$2$ holds.
Suppose, that $p$ survives the \textsc{LocalSignaling} algorithm. 
Then, there is a set $S_1 \subseteq N_{G(\ell)}(\p)$ of at least $\delta$ processes such that every process from~$S_1$ survives the first $2(\gamma-1)$ rounds of the \textsc{LocalSignaling} algorithm. Obviously, $p$ received a rumor $r$ of any process from $S_{1}$. 
By induction, for each $1\le j \le \gamma$ there is a set $S_j$ such that $S_{j-1} \subseteq S_j \subseteq N^{j}_{G(\ell)}(p)$, and all processes in~$S_j$ survive 
the first $2(\gamma-j)$ rounds of the \textsc{LocalSignaling} algorithm, and their rumors were conducted to process $p$. 
The set $S_{\gamma}$ satisfies the definition of $(\gamma,\delta)$-dense-neighborhood 
for $p$ in graph~$G|_{B_{\ell,1}}$ and the induction argument assures that the rumors of processes in $S_{\gamma}$ have reached the $p$ process.

Finally, we prove the third property.
Consider a survival set $C$ for ${B_{\ell,2}}$.
By the definition of a survival set, each process in~$C$ has at least~$\delta$ neighbors in~$C$. Because, $C \subseteq B_{\ell,2}$, thus every process from the set $C$ starts the instance of Local Signaling with variable $i$ set to $\ell$ at least. The variable $i$ decreases at most by one between every two rounds of the \textsc{LocalSignaling} algorithm, thus variables $i$ of processes in~$C$ cannot fall below $\ell$. 
In consequence, each process from $C$ that started with at the initial level $\ell$ terminates with the value $i$ being equal to $\ell$ and thus survives this instance of Local Signaling.
\end{proof}

\section{Conclusions and Open Problems}

We explored the Consensus problem in the classic message-passing model with processes' crashes, from perspective of both time {\em and} communication optimality. We discovered an interesting tradeoff between these two complexity measures: Time $\times$ Amortized\_Communication = $\logO(n)$, which, to the best of our knowledge, has not been present in other settings of Consensus and related problems.
We believe that a corresponding lower bound could be proved: Time $\times$ Amortized\_Communication = $\Tilde{\Omega}(n)$.
Interestingly, a similar tradeoff could hold between time and amount of randomness, as our main algorithm \ParamConsensus$^*$ satisfies the relation: Time $\times$ Amortized\_Randomness = $\logO(n)$. Exploring similar tradeoffs in other fault-tolerant distributed computing problems could be a promising and challenging direction to follow.

\bibliographystyle{apa}
\bibliographystyle{ACM-Reference-Format}

\end{document}